\documentclass[acmsmall,screen]{acmart}
\usepackage{tikz}
\usepackage[linesnumbered,ruled]{algorithm2e}

\usetikzlibrary{positioning,calc}
\usetikzlibrary{decorations.text}
\usetikzlibrary{decorations.pathmorphing,decorations.pathreplacing}
\usetikzlibrary{arrows,petri, topaths,fit}

\usepackage{listings}
\usepackage{wrapfig}
\usepackage{framed}   % for framing
 {\endMakeFramed}
\definecolor{shadecolor}{gray}{0.75}
\usepackage{thmtools, thm-restate}

\usepackage{multirow}
\usepackage{subcaption}
\usepackage{url}
\usepackage{graphicx}
\usepackage{tcolorbox}% http://ctan.org/pkg/tcolorbox
\usepackage{enumitem}

\newtheorem{theorem}{Theorem}[section]

\newtheorem{observation}[theorem]{Observation}
\newtheorem{claim}[theorem]{Claim}

\definecolor{mycolor}{rgb}{0.122, 0.435, 0.698}% Rule colour
\makeatletter
\newcommand{\mybox}[1]{%
  \setbox0=\hbox{#1}%
  \setlength{\@tempdima}{\dimexpr\wd0+13pt}%
  \begin{tcolorbox}[colframe=mycolor,boxrule=0.5pt,arc=4pt,
      left=6pt,right=6pt,top=3pt,bottom=3pt,boxsep=0pt,width=\@tempdima]
    #1
  \end{tcolorbox}
}

\newcommand{\introparagraph}[1]{\vspace{0.7mm} \noindent \textbf{#1.}}

\newcommand{\changes}[1]{{\color{black}{#1}}}
\newcommand{\eat}[1]{}

\newcommand{\mC}{\mathcal{C}}

\newcommand{\cfga}[1]{\mathsf{CFL^{ap}}({#1})}
\newcommand{\cfgo}[1]{\mathsf{CFL^{od}}({#1})}

\newcommand{\obtainedfrom}{\text{ :- } }

\setcopyright{rightsretained}
\acmPrice{}
\acmDOI{10.1145/3571252}
\acmYear{2023}
\copyrightyear{2023}
\acmSubmissionID{popl23main-p490-p}
\acmJournal{PACMPL}
\acmVolume{7}
\acmNumber{POPL}
\acmMonth{1}
\usepackage{booktabs}   %% For formal tables:
\usepackage{subcaption} %% For complex figures with subfigures/subcaptions
\begin{document}

%% Title information
\title{The Fine-Grained Complexity of CFL Reachability}     
\author{Paraschos Koutris}
\affiliation{
  \department{Department of Computer Sciences}              %% \department is recommended
  \institution{University of Wisconsin-Madison}
  \country{USA}
}
\email{paris@cs.wisc.edu}          %% \email is recommended

%% Author with two affiliations and emails.
\author{Shaleen Deep}
\affiliation{
  \institution{Microsoft Gray Systems Lab}
  \country{USA}
}
\email{shaleen.deep@microsoft.com}

\begin{abstract}
Many problems in static program analysis can be modeled as the context-free language (CFL) reachability problem on directed labeled graphs. The CFL reachability problem can be generally solved in time $O(n^3)$, where $n$ is the number of vertices in the graph, with some specific cases that can be solved faster. In this work, we ask the following question: given a specific CFL, what is the exact exponent in the monomial of the running time? In other words, for which cases do we have linear, quadratic or cubic algorithms, and are there problems with intermediate runtimes? This question is inspired by recent efforts to classify classic problems in terms of their exact polynomial complexity, known as {\em fine-grained complexity}. Although recent efforts have shown some conditional lower bounds (mostly for the class of combinatorial algorithms), a general picture of the fine-grained complexity landscape for CFL reachability is missing.

Our main contribution is lower bound results that pinpoint the exact running time of several classes of CFLs or specific CFLs under widely believed lower bound conjectures (Boolean Matrix Multiplication and $k$-Clique). We particularly focus on the family of Dyck-$k$ languages (which are strings with well-matched parentheses), a fundamental class of CFL reachability problems. \eat{Remarkably, we are able to show a $\Omega(n^{2.5})$ lower bound for Dyck-2 reachability, which to the best of our knowledge is the first super-quadratic lower bound that applies to all algorithms, and shows that CFL reachability is strictly harder that Boolean Matrix Multiplication.} We present new lower bounds for the case of sparse input graphs where the number of edges $m$ is the input parameter, a common setting in the database literature. For this setting, we show a cubic lower bound for Andersen's Pointer Analysis which significantly strengthens prior known results.
\end{abstract}

\maketitle

\noindent\fbox{%
    \parbox{0.98\textwidth}{%
\introparagraph{Erratum} An error was discovered in the proof of a theorem in the published version of the paper which renders the claim of a super quadratic lower bound for Dyck reachability incorrect. We thank Karl Bringmann and Marvin K\"{u}nnemann for bringing this to our attention. The claim has been removed from the current arXiv version of the paper and we note that the problem remains open. We refer the reader to~\autoref{sec:error} for a detailed description of the error.
    }
}

\section{Introduction}
\label{sec:intro}

Static analysis is the problem of approximating the run-time behaviors that a program may exhibit. It is of paramount importance in detecting bugs~\cite{bessey10,olivo}, detecting security violations and malware~\cite{jha03,livshits05}, and enabling compiler transformations and optimizations. Techniques for static analysis do not run programs on specific inputs, but instead analyze the program behavior by considering all possible inputs and executions. Since for most programs, it is impossible to go through all possible executions, it is common to use instead various approximation methods.

A standard way to express many static analysis problems is via a generalization of graph reachability called language reachability. In this setting, a directed graph $G = (V,E)$ with labeled edges from a fixed alphabet is constructed from the program code. Then, given a language $\mathcal{L}$ over the same alphabet, we seek to find pairs of nodes $s,t$ for which there is a directed path from $s$ to $t$ in $G$ such that the word formed by concatenating the labels along the path belongs to $\mathcal{L}$. An important case of language reachability, which is going to be the main subject in this work, is {\em CFL reachability}, where $\mathcal{L}$  is a context-free language~\cite{Reps98,Yannakakis90}.
The CFL reachability problem has applications to a wide range of static analysis problems, including interprocedural data-flow analysis~\cite{RepsHS95}, shape analysis~\cite{Reps95}, type-based flow analysis~\cite{RehofF01}, and points-to analysis~\cite{ShangXX12, ZhengR08}. CFL reachability is also an important problem in database theory, since it is equivalent to a class of Datalog programs called {\em chain Datalog programs}~\cite{SmaragdakisB15,Reps98}, where the bodies of the recursive rules form a chain of binary predicates.

%In most of the above cases, the languages used to define the problem are those of properly-matched parenthesis, which are known as Dyck languages, and form a proper subset of context-free languages.

%Datalog is a recursive query language that has seen a resurgence over the last decade, with applications in graph analytics~\cite{SeoPSL13}, static analysis of programs~\cite{Souffle,SmaragdakisB15}, and business analytics~\cite{GreenHLZ13}. A long line of research has studied the complexity of Datalog evaluation. 
%Certain fragments of Datalog have lower data complexity. Evaluation for non-recursive Datalog is in $\mathsf{AC}_0$, while evaluation for Datalog with linear rules is in $\mathsf{NC}$ and thus parallelizable~\cite{UllmanG88}. 

\introparagraph{The Complexity of CFL reachability} It was shown by Yannakakis~\cite{Yannakakis90} that the general CFL reachability problem can be solved in $O(n^3)$ time on general graphs for a fixed language, where $n$ is the number of vertices in the input graph. This runtime has only been slightly improved by a logarithmic factor to $O(n^3/\log n)$~\cite{Chaudhuri08} for the general case. Some improvements exist for more restricted languages: for example, regular languages admit an $O(n^\omega)$ algorithm~\cite{fischer1971boolean}, where $\omega$ is the matrix multiplication exponent (the current best known value is $\omega \approx 2.37$). \changes{In fact, a simple algorithm (with a slightly worse running time) can be obtained as follows. A graph can be encoded as a square matrix $A(i,j)$ and using the observation that $A^2$ encodes the number of walks from vertex $i$ to $j$ of length $2$, one can perform a logarithmic number of iterated multiplications to double the walk path length in each step and obtain an $O(n^\omega \log n)$ algorithm. However, algorithms based on fast matrix multiplication are not desirable practically since they hide large constants in the big-$O$ running time. Currently all known truly sub-cubic algorithms use matrix multiplication}. This lack of progress has lead to a conjecture that no better
algorithm exists for CFL reachability, that is, a $O(n^{3-\epsilon})$ runtime is not possible for any constant $\epsilon >0$. This conjecture has been (conditionally) proven but only for the class of combinatorial algorithms~\cite{CCP18,ABW18}. Combinatorial algorithms are algorithms with a small constant in the big-$O$ that can be implemented efficiently in practice. In contrast, non-combinatorial algorithms can use algebraic methods such as fast matrix multiplication, which potentially obtain a faster theoretical runtime but are not practical.

However, these results do not tell us how efficiently we can evaluate CFL reachability for a specific language $\mathcal{L}$.
This motivates us to take a different approach. Given a CFL (or a context-free grammar -- CFG), we ask to {\em  identify the exact expression of the running time as a function of the input size}. For example, which languages run in linear time, and for which programs do we need quadratic or cubic time? Answering such a question is important since CFL reachability can naturally capture several fundamental computational problems that are in \textsf{P}. We consider two variants of the CFL reachability problem: in the \textsc{All-Pairs} problem, we produce all pairs that are reachable, while in the \textsc{On-Demand} problem we check reachability for a given pair of vertices. 
\begin{sloppypar}
\introparagraph{Fine-grained Complexity}
The research direction of pinpointing the exact running time of the problems as a function of their input size is related to the area of {\em fine-grained complexity} \cite{williams2018subcubic}. Since obtaining unconditional lower bounds for polynomial running times is not within our reach, the goal of fine-grained complexity is to reduce a given problem in \textsf{P} to one of a small set of problems that are widely believed to have an optimal algorithm (e.g.,  3-SUM, Boolean Matrix Multiplication (BMM, for short), $k$-Clique). Our goal in this paper follows the same general direction: we seek to show that the complexity for a given language is optimal conditional to one of these conjectures. 
\end{sloppypar}
\introparagraph{The Case of Sparse Graphs} 
In the CFL reachability problem, the most common parameter used as input size is the number of vertices $n = |V|$. In this case, the lower bounds usually construct instances that are dense (in the sense that the number of edges is super-linear or even quadratic with respect to $n$). However, in many practical instances the graph $G$ constructed is sparse, and thus it is meaningful to use as input size the number of edges $m = |E|$. This is also the case when we view CFL reachability from a database lens, since $m$ translates to the size of the database (i.e., the number of tuples across all relations). In this work, we will state our fine-grained complexity results using both parameters.  

\subsection{Our Contributions}

\introparagraph{Dyck Reachability} We first study the fine-grained complexity of a fundamental class of CFL reachability problems, called Dyck reachability (Section~\ref{sec:dyck}). The Dyck-$k$ grammar produces words of well-matched parentheses of $k$ different types. \changes{When we restrict to combinatorial algorithms, it is known that a conditional cubic lower bound (wrt to parameter $n$) exists for the \textsc{On-Demand} problem for Dyck-$k$ for any $k \geq 1$~\cite{CCP18,zhang2020conditional,hansen2021tight}.} We show that the conditional cubic lower bound also exists in the \emph{sparse} setting, i.e. the cubic lower bound continues to exist wrt to parameter $m$, for the \textsc{On-Demand} problem for Dyck-$k$ for any $k \geq 2$.

\eat{We show that there is no algorithm that solves \textsc{All-Pairs} Dyck-2 reachability that requires $O(n^{2.5-\epsilon})$ time for any $\epsilon >0$ (Theorem~\ref{thm:bound}). This lower bound uses a reduction from the APSP or 3SUM hypothesis and thus applies to {\em all algorithms}, even non-combinatorial ones. Since matrix multiplication is done in time $O(n^\omega)$ and $\omega < 2.5$, this shows that the general CFL reachability problem is likely strictly harder than Boolean Matrix Multiplication. To the best of our knowledge, this is the first result that shows a super-quadratic lower bound for CFL reachability that applies to all algorithms.}

\introparagraph{All-Pairs CFL Reachability} Our next set of results looks at general context-free grammars (Section~\ref{sec:cfg}). We identify a syntactic condition that is checkable in polynomial time w.r.t. the size of the grammar, such that any CFG that satisfies this condition is as hard as BMM; otherwise, it can be solved in time $O(m)$. Since the combinatorial BMM hypothesis says that BMM cannot be done faster than $O(n^{3-\epsilon})$, this implies a surprising classification result in the combinatorial setting (Theorem~\ref{thm:dichotomy}): we can say exactly for which CFGs the all-pairs CFL Reachability problem can be solved in optimal time $O(n^3)$, and for which in optimal time $O(n^2)$. In other words, there exists a sharp dichotomy in the runtime, with no in-between exponents in the polynomial. In the non-combinatorial setting this dichotomy disappears, and we can identify problems with intermediate running times of $O(n^\omega)$ (when CFG is regular), and $O(n^{(3+\omega)/2})$.  We also show that identifying for a given language its exact exponent is actually an undecidable problem (although this does not exclude a possible dichotomy result with an undecidable syntactic condition). 

\introparagraph{On-Demand  CFL Reachability} Next we turn our attention to the easier problem of on-demand CFL Reachability (Section~\ref{sec:on-demand}). We sketch the fine-grained complexity landscape for both dense and sparse graphs, and provide new conditional lower bounds for several interesting CFGs. A summary of our results can be shown in Table~\ref{table:2}. Interestingly, all runtimes we have identified for combinatorial algorithms are either linear, quadratic, or cubic to the input size, so it is an intriguing question whether other intermediate exponents are possible.

\introparagraph{Andersen's Pointer Analysis} Finally, we look at the fine-grained complexity for the Andersen's Pointer Analysis (APA), a fundamental type of points-to analysis (Section~\ref{sec:binary}). Although APA is not captured directly as a CFL reachability problem, we can slightly rewrite the program so that it behaves as one. In this way, we can use our techniques to show a lower bound of $O(m^{3-\epsilon})$ for any $\epsilon >0$ under the combinatorial $k$-Clique hypothesis, even applying to the on-demand setting.  So far a cubic lower bound was only known with respect to $n$~\cite{Andersen}, so this is a significant strengthening of the lower bound to sparse inputs.

\section{Preliminaries}
\label{sec:prelim}

\introparagraph{Context-Free Grammars} 
A {\em context-free grammar (CFG)} $\mathcal{G}$ can be described by a tuple $(V, \Sigma, R, S)$, where: $V$ is a finite set of variables (which are non-terminal), $\Sigma$ is a finite set (disjoint from $V$) of terminal symbols, $R$ is a set of production rules where each production rule maps a variable to a string  $\in (V \cup \Sigma)^*$, and $S$ is a start symbol from $V$. For example, the grammar $\{ S \gets \epsilon, S \gets a S b \}$ is a CFG that describes all strings of the form $a^ib^i$ for some $i \geq 0$. A {\em context-free language (CFL)} is a language that is produced by some CFG. We will denote by $L(\mathcal{G})$ the language produced by $\mathcal{G}$.

A CFG is {\em right-regular} if all productions rules are of the form $S \gets \epsilon, S \gets a$ or $S \gets a B$. We can similarly define a left-regular CFG.  Right-regular (or left-regular) grammars generate exactly all regular languages.
A CFG is {\em linear} if every production rule contains at most one non-terminal symbol in its body. 
 For example, the grammar $\{ S \gets \epsilon, S \gets a S b \}$ is linear.

\introparagraph{Dyck-$k$ Grammars}
Of particular interest to us will be the family of Dyck-$k$ grammars, which are not regular and linear. The Dyck-$k$ grammar $\mathcal{D}_k$ captures the language of strings with well-matched parentheses of $k$ different types.
\begin{align*}
S & \leftarrow \epsilon \mid a_1 S \bar{a}_1  \mid a_2 S \bar{a}_2  \mid \dots \mid a_k S \bar{a}_k 
\end{align*}

\introparagraph{CFL Reachability} The CFL reachability problem takes as an input a directed graph $G = (V,E)$ whose edges are labelled by an alphabet $\Sigma$, and a CFG $\mathcal{G}$ defined over the same alphabet $\Sigma$. We say that a vertex $v \in V$ is $L$-reachable from a vertex $u \in V$ if there is a path from $u$ to $v$ in $G$ such that the labels of the edges form a string that belongs in the language $L$. We consider two variants of the CFL reachability problem:
\begin{itemize}
\item $\textsc{All-Pairs}$: output all pairs of vertices $u, v$  such that $v$ is $L(\mathcal{G})$-reachable from $u$ in $G$.
\item $\textsc{On-Demand}$: given a pair of vertices $u,v$, check whether  $v$ is $L(\mathcal{G})$-reachable from $u$ in $G$.
%\item $\textsc{Boolean}(P)$: decide whether there exists constants $a,b$ such that $T(a,b)$ is true.
\end{itemize}

\introparagraph{Complexity Problems}
In this paper, we will consider the grammar $\mathcal{G}$ as being fixed (i.e., of constant size), and we will be interested in the complexity of CFL reachability with input the graph $G$. Following this, we define as $\cfga{\mathcal{G}}$ the \textsc{All-Pairs} problem for a fixed grammar $\mathcal{G}$, and as  $\cfgo{\mathcal{G}}$ the \textsc{On-Demand} problem for a fixed grammar $\mathcal{G}$. We then ask the following question: how does the grammar $\mathcal{G}$ effect the computational complexity of CFL reachability problem? This deviates from most previous approaches, which were interested in the computational complexity across all possible CFGs.

\introparagraph{Parameters}
The input to both problems is a graph $G = (V,E)$. To measure the complexity of CFL reachability, we will use as parameters both the number of vertices $n = |V|$ and the number of edges $m = |E|$. As we will see over the next sections, our results differ depending on the parameter we focus on. To simplify our presentation, we will assume w.l.o.g. that $G$ does not have any isolated vertices (i.e., vertices without adjacent edges).  Isolated vertices can only help to satisfy the empty string (if the CFG accepts it), and hence can be handled in time $O(n)$ and then removed from $G$. This will only add a linear term w.r.t. $n$ in the running time of any algorithm, which we will thus ignore when we measure complexity w.r.t. $m$. Hence, we will use throughout the paper the following inequality: $n/2 \leq m \leq n^2$. We will also use the notation $V_G$ and $E_G$ to denote the vertex set $V$ and edge set $E$ corresponding to the graph $G$.

%Recall that a Datalog chain program corresponds to a CFG (without the empty word) and vice versa. It will be convenient to switch between these two formalisms. Let $P$ be a chain Datalog program with target the binary predicate $T$. We will consider two different modes of evaluating $P$, which are useful for different practical scenarios:

\introparagraph{Computational Model} We will consider the word-RAM model with $O(\log n)$ bit words. This is a RAM machine that
can read from memory, write to memory and perform operations on $O(\log n)$ bit blocks of data in constant time.

\introparagraph{Combinatorial Algorithms} 
In this work, we will often restrict our attention to {\em combinatorial} algorithms~\cite{williams2018subcubic}. This notion is not precisely defined, but informally, it means that the algorithm is discrete, graph-theoretic, and with a runtime has a small constant in the big-$O$. \changes{This requirement disallows the use of fast matrix multiplication, including Strassen's algorithm~\cite{strassen1969gaussian}.} The notion of combinatorial algorithms is used to distinguish them from  {\em algebraic} algorithms, the most common example of these being the subcubic algorithms that multiply two boolean $n \times n$ matrices in time $O(n^\omega)$ with $\omega < 3$.

\changes{\subsection{Fine-Grained Complexity}

Fine-grained complexity is a powerful tool to reason about lower bounds for problems solvable in polynomial time. Consider a problem $A$ with input size $n$. If $A$ can be solved in polynomial time, our goal is to find the smallest constant $c>0$ such that $A$ can be solved in time $O(n^c)$

Let us consider one of the simplest problems that have widespread use in fine-grained complexity. The 3SUM problem asks whether, given $n$ integers, three integers exist that sum to $0$. There exists a straightforward algorithm that solves 3SUM in quadratic time. However, despite decades of research, it remains unknown if there exists a sub-quadratic time algorithm, i.e., is there an algorithm that takes time $O(n^{2 - \epsilon})$ for some constant $\epsilon > 0$. The first step in fine-grained complexity is establishing reasonable conjectures about the running times for well-studied computational problems. In this paper, we will use the following well-established conjectures to prove our conditional lower bounds:
\begin{description}
	\item[ 3SUM hypothesis~\cite{gajentaan1995class}:] There is no $O(n^{2-\epsilon})$ time algorithm for 3SUM, for any constant $\epsilon > 0$. The 3SUM problem takes as input $n$ integers in $\{-n^c, \dots, n^c\}$ for a constant $c$, and asks whether any three of the integers sum to 0.
	\item[APSP hypothesis~\cite{williams2018subcubic}:] There is no $O(n^{3-\epsilon})$ time algorithm for the All-Pairs Shortest Path problem, for any constant $\epsilon > 0$. 
	\item[Combinatorial BMM hypothesis:] There is no combinatorial algorithm that can solve Boolean Matrix Multiplication on boolean matrices of dimensions $n \times n$ with running time $O(n^{3-\epsilon})$ for any constant $\epsilon > 0$.
	\item[Combinatorial $k$-Clique hypothesis:] For any  $k \geq 3$, there is no combinatorial algorithm that detects a $k$-Clique in a graph with $n$ nodes in time $O(n^{k-\epsilon})$ for any constant $\epsilon > 0$.
\end{description}

The Combinatorial $k$-Clique hypothesis is a generalization of the Combinatorial BMM hypothesis, since combinatorial BMM is equivalent to combinatorial triangle (clique with $k=3$) detection~\cite{williams2018subcubic}.

The second step in fine-grained complexity is to reason about \emph{fine-grained reductions}. Suppose we have a problem $A$ with running time $a(n)$ and a problem $B$ with running time $b(n)$. Given an oracle that can solve problem $B$ in time $O(b(n)^{1-\epsilon})$ for some $\epsilon > 0$, we would like to somehow use this oracle to obtain an algorithm for problem $A$ with running time $O(a(n)^{1 - \epsilon'})$ for some $\epsilon' > 0$. The transformation of instances of $A$ to instances of $B$ (aka a reduction) requires some special properties. It is not enough to have a polynomial time reduction from $A$ to $B$. Instead, we want to ensure that the reduction runs in time faster than $a(n)$. Further, we also want the ability to make multiple calls to the oracle for $B$ (aka a Turing-style reduction). 

%These requirements can now be formalized as follows.
%
%\begin{definition}\cite{williams2018some} \label{def:fg}
%	Consider two problems $A$ and $B$ with conjectured running time lower bounds as $a(n)$ and $b(n)$ respectively. Then, $A$ can be reduced to $B$ with parameters $(a,b)$, denoted $A \leq_{a,b} B$, if for every $\epsilon > 0$, there exists a $\delta > 0$, and an algorithm $R$ (the reduction) that runs in time $a(n)^{1- \delta}$ on some input of length $n$. Algorithm $R$ is allowed to make $q$ calls to the oracle for $B$ with input sizes $n_1, \dots n_q$, such that,
%	
%	$$ \sum_{i=1}^q (b(n_i))^{1- \epsilon} \leq a(n)^{1 - \delta}  $$
%\end{definition}
%
%If $A \leq_{a,b} B$, then \autoref{def:fg} implies that if $B$ has an algorithm with running time $O(b(n)^{1 - \epsilon})$, then $A$ can be solved in time $O(a(n)^{1 - \delta})$ for some $\epsilon, \delta > 0$. Therefore, if we want to reason about lower bounds for $A$, then we need to find a suitable problem $B$ with a believable running time conjecture, and perform a fine-grained reduction. Note that if $A \leq_{a,b} B$ and $B \leq_{a,b} A$, then $A$ and $B$ are considered to be fine-grained equivalent. 
\begin{sloppypar}
Fine-grained complexity based lower bounds have been a useful yardstick to understand the hardness for many problems solvable in polynomial time where progress has been stalled for several years (and in some cases, decades). This is not limited to static problems. Seminal work~\cite{henzinger2015unifying} has also proposed fundamental conjectures for dynamic problems that have been used  to reason about the optimality of algorithms in the dynamic setting.
\end{sloppypar}
}

\eat{\introparagraph{Lower Bounds Hypotheses} We will use the following well-established conjectures to prove our conditional lower bounds:
\begin{description}
\item[ 3SUM hypothesis:] There is no $O(n^{2-\epsilon})$ time algorithm for 3SUM, for any constant $\epsilon > 0$. The 3SUM problem takes as input $n$ integers in $\{-n^c, \dots, n^c\}$ for constant $c$, and asks if three of the integers sum to 0.
\item[APSP hypothesis:] There is no $O(n^{3-\epsilon})$ time algorithm for the All-Pairs Shortest Path problem, for any constant $\epsilon > 0$. 
\item[Combinatorial BMM hypothesis:] There is no combinatorial algorithm that can solve Boolean Matrix Multiplication on boolean matrices of dimensions $n \times n$ with running time $O(n^{3-\epsilon})$ for any constant $\epsilon > 0$.
\item[Combinatorial $k$-Clique hypothesis:] For any  $k \geq 3$, there is no combinatorial algorithm that detects a $k$-Clique in a graph with $n$ nodes in time $O(n^{k-\epsilon})$ for any constant $\epsilon > 0$.
\end{description}}

\section{Dyck Reachability}
\label{sec:dyck}

In this section, we study the running time of the \textsc{All-Pairs} and \textsc{On-Demand} CFL reachability problems for the family of Dyck-$k$ grammars. We focus on this family of grammars for two reasons. First, Dyck reachability is a fundamental  problem at the heart of static analysis. Second,  Dyck-2 is in some sense the "hardest" CFG~\cite{Greibach73}, so its complexity will be informative of the behavior of other CFGs. Here, we should note here that any Dyck-$k$ problem for $k \geq 2$ is equivalent (w.r.t. running time) to Dyck-2.\footnote{One can encode $k$ types of parentheses with 2 types using a simple binary encoding.}

We begin by recalling a result for the \textsc{On-Demand} problem on $\mathcal{D}_k$. This result was proven in a non peer-reviewed publication~\cite{zhang2020conditional} and~\cite{hansen2021tight}.

\begin{restatable}{theorem}{cflod}~\cite{zhang2020conditional,hansen2021tight} \label{thm:dyck:n}
	Under the combinatorial BMM hypothesis, there is no combinatorial algorithm that evaluates $\cfgo{\mathcal{D}_k}$ for $k \geq 1$ in $O(n^{3-\epsilon})$ for any constant $\epsilon > 0$.
\end{restatable}

Next, we consider the case of sparse graphs, where we are only interested in the number of edges $m$ as the input parameter. We can show the following result.

\begin{restatable}{theorem}{dyck} \label{thm:dyck}
Under the combinatorial $k$-Clique hypothesis,  $\cfgo{\mathcal{D}_k}$ (and thus $\cfga{\mathcal{D}_k}$) for $k \geq 2$ cannot be solved by a combinatorial algorithm in time $O(m^{3-\epsilon})$ for any constant $\epsilon > 0$. 
\end{restatable}

Compared to Theorem~\ref{thm:dyck:n}, the above lower bound is stronger, but it is based on a weaker hypothesis, since combinatorial BMM is equivalent to combinatorial 3-Clique. Moreover, Theorem~\ref{thm:dyck} does not apply to Dyck-1; we do not know whether Dyck-1 admits a faster algorithm on sparse inputs. We should note that this lower bound was also shown in~\cite{schepper2018complexity}, but indirectly via a reduction from a problem in pushdown automata. Our simplified construction allows us to reuse the same gadgets for proving the lower bounds for Andersen's Pointer Analysis in Section~\ref{sec:binary}. We present the proof next.

\subsection{Proof of Theorem~\ref{thm:dyck}}

\begin{figure}[!t]
	%\hspace*{-9em}
	\scalebox{0.7}{
	\begin{tikzpicture}
		\def\cl1{1}
		\def\cng1{2.2}
		\def\cngfig{3.85}
			\tikzstyle{every node}=[font=\Large]
			%% first CL1 gadget %%%%%%%%%%%%%%%%
			\node (v1)  at (-8.5, 3.75) { $\blacksquare$};
			\node (v2)  at (-7.5, 3.75){ $\blacksquare$};
			\node (d1)  at (-6.5, 3.75) { $\dots$};
			\node (vk)  at (-5.5, 3.75) { $\blacksquare$};
			
			\draw[->] (v1) -- (v2) node[midway,above]{$v_1$} ;
			\draw[->] (v2) -- (d1) node[midway,above]{$v_2$} ;
			\draw[->] (d1) -- (vk);
			
			\draw [decorate,decoration = {brace, mirror}] (-8.5,3.5) --  (-4.6,3.5);
			\node at (-6.5, 3.2) {\small $CL_1(t_1)$};
			%% first CL1 gadget %%%%%%%%%%%%%%%%
			
			\node at (-6.5, 1.75) { $\vdots$};
			
			%% last CL1 gadget %%%%%%%%%%%%%%%%
			\node (v1k)  at (-8.5, 0.9-\cl1) { $\blacksquare$};
			\node (v2k)  at (-7.5, 0.9-\cl1) { $\blacksquare$};
			\node (d1k)  at (-6.5, 0.9-\cl1) { $\dots$};
			\node (vkk)  at (-5.5, 0.9-\cl1) { $\blacksquare$};
			
			\draw[->] (v1k) -- (v2k) node[midway,above]{$u_1$} ;
			\draw[->] (v2k) -- (d1k) node[midway,above]{$u_2$} ;
			\draw[->] (d1k) -- (vkk);
			
			\draw [decorate,decoration = {brace, mirror}] (-8.5,0.65-\cl1) --  (-4.6,0.65-\cl1);
			\node at (-6.5, 0.3-\cl1) {\small $CL_1(t_{|\mC_k|})$};
			%% last CL1 gadget %%%%%%%%%%%%%%%%
			
			\draw [decorate,decoration = {brace, mirror}] (-4.5,2.5) --  (-0.5,2.5);
			\node at (-2.35, 2.15) {\small $CNG_1(t_1)$};
			
			\draw [decorate,
			decoration = {brace, mirror}] (-4.5,0.75-\cng1) --  (-0.5,0.75-\cng1);
			\node at (-2.35, 0.35-\cng1) {\small $CNG_1(t_{|\mC_k|})$};
			
			\node (p) at (-9.5, 1.75) { $p$};
			\draw[->] (p) |- (v1);
			\draw[->] (p) |- (v1k);
			
			%\draw [decorate,decoration = {brace, mirror}] (-9, 5) --  (-9, -1.5);
			%\draw [decorate,decoration = {brace}] (0, 5) --  (0, -1.5);

			\node (center1) at ( -4.5, 3.75) {$\blacksquare$};
			
			\draw[->] (vk) -- (center1) node[midway,above]{$v_k$} ;
			
			\node (x1)  at (-4, 5) { $x_1$};
			\node (x2)  at (-4, 4.5) { $x_2$};
			\node (xp)  at (-4, 2.75) { $x_\ell$};
			\node (sp1) at (-3.5, 3.75) {$\blacksquare$};
			\node (sp2) at (-2.5, 3.75) {$\blacksquare$};
			\node (sp3) at (-1.5, 3.75) {$\blacksquare$};
			
			\draw[->] (center1) |- (x1);
			\draw[->] (center1) |- (x2);
			\draw[->] (center1) |- (xp);
			\draw[->] (x1) -| ($(sp1.north)-(2pt,0)$);
			\draw[->] (x2) -| ($(sp1.north)-(2pt,0)$);
			\draw[->] (xp) -| ($(sp1.south)-(2pt,0)$);

			\node (x11)  at (-3, 5) { $x_1$};
			\node (x111) at (-2, 5) { $\dots$};
			\node  (x1111) at (-1, 5) { $x_1$};
			
			\node (x21)  at (-3, 4.5) { $x_2$};
			\node (x211) at (-2, 4.5) { $\dots$};
			\node  (x2111) at (-1, 4.5) { $x_2$};
			
			\node (xp1)  at (-3, 2.75) { $x_\ell$};
			\node (xp11) at (-2, 2.75) { $\dots$};
			\node  (xp111) at (-1, 2.75) { $x_\ell$};
			
			\node (x3)  at (-4, 3.75) {$\vdots$};
			\node (x31)  at (-3, 3.75) {$\vdots$};
			\node (d4)  at (-2, 3.75) { $\vdots$};
			\node (x311)  at (-1, 3.75) {$\vdots$};
			
			\draw[->] ($(sp1.north)+(2pt,0)$) |- (x11);
			\draw[->] ($(sp1.north)+(2pt,0)$) |- (x21);
			\draw[->] ($(sp1.south)+(2pt,0)$) |- (xp1);
			\draw[->] (x11) -| (sp2);
			\draw[->] (x21) -| (sp2);
			\draw[->] (xp1) -| (sp2);
			\draw[->] (sp3) |- (x1111);
			\draw[->] (sp3) |- (x2111);
			\draw[->] (sp3) |- (xp111);
			
			\node[circle, draw, minimum size=10pt,inner sep=1pt, outer sep=2pt] (center2) at ( 0.25, 1.75) {$A$};
			
			\draw[->] (x1111) -| ($(center2.north)-(3pt,0)$);		
			\draw[->] (x2111) -| ($(center2.north)-(3pt,0)$);	
			\draw[->] (xp111) -| ($(center2.north)-(3pt,0)$);	
			
			\node (center1k) at ( -4.5, 3.75-\cngfig) {$\blacksquare$};
			
			\draw[->] (vkk) -- (center1k) node[midway,above]{$u_k$} ;
			
			\node (x1k)  at (-4, 5-\cngfig) { $x'_1$};
			\node (x2k)  at (-4, 4.5-\cngfig) { $x'_2$};
			\node (d4k)  at (-2, 3.75-\cngfig) { $\vdots$};
			\node (xpk)  at (-4, 2.75-\cngfig) { $x'_\ell$};
			\node (sp1k) at (-3.5, 3.75-\cngfig) {$\blacksquare$};
			\node (sp2k) at (-2.5, 3.75-\cngfig) {$\blacksquare$};
			\node (sp3k) at (-1.5, 3.75-\cngfig) {$\blacksquare$};
			\node  at (-4, 3.75-\cngfig) { $\vdots$};
			\node  at (-3, 3.75-\cngfig) { $\vdots$};
			\node  at (-1, 3.75-\cngfig) { $\vdots$};
			
			\draw[->] (center1k) |- (x1k);
			\draw[->] (center1k) |- (x2k);
			\draw[->] (center1k) |- (xpk);
			\draw[->] (x1k) -| ($(sp1k.north)-(2pt,0)$);
			\draw[->] (x2k) -| ($(sp1k.north)-(2pt,0)$);
			\draw[->] (xpk) -| ($(sp1k.south)-(2pt,0)$);

			\node (x11k)  at (-3, 5-\cngfig) { $x'_1$};
			\node (x111k) at (-2, 5-\cngfig) { $\dots$};
			\node  (x1111k) at (-1, 5-\cngfig) { $x'_1$};

			\node (x21k)  at (-3, 4.5-\cngfig) { $x'_2$};
			\node (x211k) at (-2, 4.5-\cngfig) { $\dots$};
			\node  (x2111k) at (-1, 4.5-\cngfig) { $x'_2$};
			
			\node (xp1k)  at (-3, 2.75-\cngfig) { $x'_\ell$};
			\node (xp11k) at (-2, 2.75-\cngfig) { $\dots$};
			\node  (xp111k) at (-1, 2.75-\cngfig) { $x'_\ell$};
			
			\draw[->] ($(sp1k.north)+(2pt,0)$) |- (x11k);
			\draw[->] ($(sp1k.north)+(2pt,0)$) |- (x21k);
			\draw[->] ($(sp1k.south)+(2pt,0)$) |- (xp1k);
			\draw[->] (x11k) -| (sp2k);
			\draw[->] (x21k) -| (sp2k);
			\draw[->] (xp1k) -| (sp2k);
			\draw[->] (sp3k) |- (x1111k);
			\draw[->] (sp3k) |- (x2111k);
			\draw[->] (sp3k) |- (xp111k);
			\draw[->] (x1111k) -| ($(center2.south)-(3pt,0)$);		\draw[->] (x2111k) -| ($(center2.south)-(3pt,0)$); \draw[->] (xp111k) -| ($(center2.south)-(3pt,0)$);

			\def\off{1}
			\node[circle, draw, minimum size=10pt, inner sep=1pt, outer sep=2pt] (center5) at ( 5.25, 1.75) {$B$};
			
			%% Second bracket %%%%%%%%%%%%%%%%%%%%%%%%%%%			
			\node (v12)  at (1, 3.75) {$\blacksquare$};
			\node (d12)  at (1+\off, 3.75) { \small $\dots$};
			\node (center4) at ( 1+\off+\off, 3.75) {$\blacksquare$};
			
			\draw[->] (v12) -- (d12) node[midway,above]{$v_{k-1}^R$} ;
			\draw[->] (d12) -- (center4) node[midway,above]{$v_1^R$} ;
			\draw[->] ($(center2.north)+(3pt,0)$)	 |- (v12) node[midway,below, xshift = 6pt, yshift = -5pt]{$v_{k}^R$} ;
			
			\draw [decorate,decoration = {brace, mirror}] (1,3.5) --  (2.95,3.5);
			\node at (2.0, 3.2) {\small $\overline{CL}_2(t_1)$};
			
			\draw [decorate,decoration = {brace, mirror}] (3.05,3.5) --  (5,3.5);
			\node at (4, 3.2) {\small ${CNG}_2(t_1)$};
			\node   at (4, 3.75) { $\dots$};
			\node   at (2.75, 2) { $\vdots$};

			\draw[->] (4.75, 3.75) -| ($(center5.north)-(3pt,0)$);
			\draw[->] (4.75, 3.75 - \cngfig) -| ($(center5.south)-(3pt,0)$);

			\node (v12k)  at (1, 3.75 - \cngfig) {$\blacksquare$};
			\node (d12k)  at (1+\off, 3.75- \cngfig) { \small $\dots$};
			%\node (vk2k)  at (2.5, 3.75- \cngfig) { \small $v_1$};
			\node(center4k) at ( 1+\off+\off, 3.75- \cngfig) {$\blacksquare$};
			
			\draw[->] (v12k) -- (d12k) node[midway,above]{$u_{k-1}^R$} ;
			\draw[->] (d12k) -- (center4k) node[midway,above]{$u_1^R$} ;
			\draw[->] ($(center2.south)+(3pt,0)$) |- (v12k)
			node[midway,above, xshift = 6pt, yshift = 10pt]{$u_{k}^R$} ;
			
			\draw [decorate,
			decoration = {brace, mirror}] (1,3.5- \cngfig) --  (2.95,3.5- \cngfig);
			\node at (2.0, 3.2- \cngfig) {\small $\overline{CL}_2(t_{|\mC_k|})$};
			\draw [decorate,
			decoration = {brace, mirror}] (3.05,3.5- \cngfig) --  (5,3.5- \cngfig);
			\node at (4, 3.2- \cngfig) {\small ${CNG}_2(t_{|\mC_k|})$};
			\node   at (4, 3.75- \cngfig) { $\dots$};

			%%%%content inside third brace
			\node (v13)  at (6, 3.75) {$\square$};
			\node (d13)  at (6+\off, 3.75) { \small $\dots$};
			%\node (vk3)  at (7.5, 3.75){$\square$};
			\node (center6) at (6+\off+\off, 3.75) {$\blacksquare$};
			
			\draw[->] (v13) -- (d13) node[midway,above]{$v_{k-1}^R$} ;
			\draw[->] (d13) -- (center6) node[midway,above]{$v_1^R$} ;
			\draw[->] ($(center5.north)+(3pt,0)$) |- (v13)
			node[midway,below, xshift = 6pt, yshift = -5pt]{$v_{k}^R$} ;
			
			\draw [decorate,
			decoration = {brace, mirror}] (6,3.5) --  (7.95,3.5);
			\node at (7.0, 3.2) {\small $\overline{CL}_3(t_1)$};
			\draw [decorate,
			decoration = {brace, mirror}] (8.05,3.5) --  (10,3.5);
			\node at (9, 3.2) {\small $\overline{CNG}_3(t_1)$};
			\node   at (9, 3.75) { $\dots$};
			\node   at (7.75, 2) { $\vdots$};
			
			\node (center7) at ( 10.5, 1.75) {$q$};
			
			\draw[->] (9.95, 3.75) -| (center7);
			\draw[->] (9.95, 3.75 - \cngfig) -| (center7);

			\node (v13k)  at (6, 3.75  - \cngfig) {$\square$};
			\node (d13k)  at (6+\off, 3.75- \cngfig) { \small $\dots$};
			\node(center6k) at ( 8, 3.75- \cngfig) {$\blacksquare$};
			%\node[circle, fill, color=blue, minimum size=1pt] (center6k) at ( 8, 3.75- \cngfig) {};
			
			\draw[->] (v13k) -- (d13k) node[midway,above]{$u_{k-1}^R$} ;
			\draw[->] (d13k) -- (center6k) node[midway,above]{$u_1^R$} ;
			%\draw[->] (vk3k) -- (center6k);
			\draw[->] ($(center5.south)+(3pt,0)$) |- (v13k)
			node[midway,above, xshift = 6pt, yshift = 10pt]{$u_{k}^R$} ;
			
			\draw [decorate,
			decoration = {brace, mirror}] (6,3.5- \cngfig) --  (7.95,3.5- \cngfig);
			\node at (7.0, 3.2- \cngfig) {\small $\overline{CL}_3(t_{|\mC_k|})$};
			\draw [decorate,
			decoration = {brace, mirror}] (8.05,3.5- \cngfig) --  (10,3.5- \cngfig);
			\node at (9.1, 3.175- \cngfig) {\small $\overline{CNG}_3(t_{|\mC_k|})$};
			\node   at (9, 3.75- \cngfig) { $\dots$};
		
	\end{tikzpicture}}
	\caption{The input graph constructed for the \textsc{On-Demand} Dyck-$2$. We are checking whether $(p,q)$ is in the output. For the clique $t_1$, $N_{t_1} = \{x_1, \dots, x_\ell\}$ and $CNG_1(t_1)$ uses $k$ copies of $N_{t_1}$. } \label{fig:dyck}
\end{figure}
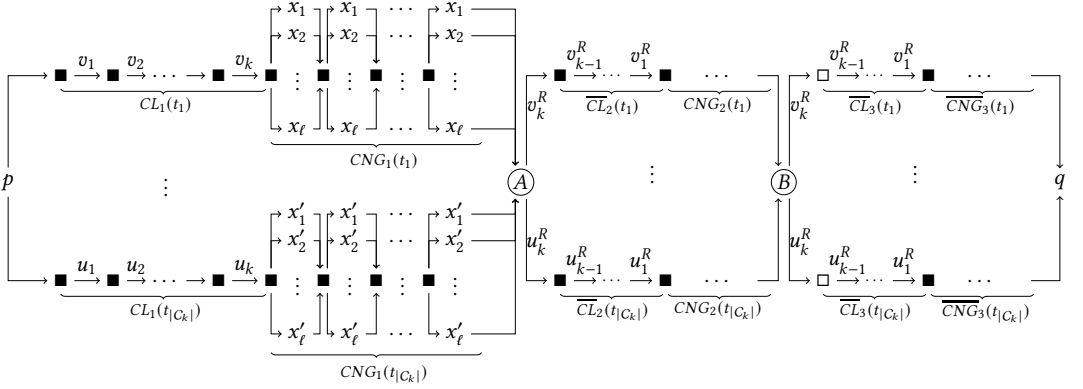

The reduction uses the same core idea as the one used in~\cite{ABW18}. The construction is based on the following idea: If there is a $3k$-clique, then there are 3 disjoint $k$-cliques. Moreover, if every pair of these 3 cliques forms a $2k$-clique, then there is a $3$k-clique in the graph. To reduce the $3k$-Clique problem, we will take as input a graph $G$ with $n$ nodes and transform it to an instance for Dyck-$2$ with $O(n^{k+1} \log n)$ edges. This transformation will be done in time $O(n^{k+1} \log n)$ time. We obtain the desired bound by letting $k$ grow depending on the constant $\epsilon$.

\paragraph{Notation} We associate with each node in the graph an integer in $\{1, \dots, n \}$. Let $\bar{v}$ denote the binary encoding of a node $v$ and let $\bar{v}^R$ denote the reverse sequence of $\bar{v}$. We will assume that the length of the binary encoding $\bar{v}$ is exactly $N = \log n$ (we can always pad with 0's). Since we will construct a Dyck-$2$ instance, bits $0, 1$ in $\bar{v}$ will be encoded using symbols $[$  and $($ respectively and the bits in $\bar{v}^R$ will be encoded using $]$ and $)$ instead of $0$ and $1$ respectively.  

The numbers $\bar{v}, \bar{v}^R$ will be encoded as directed line graphs with $N$ edges. The label of the $i^{\text{th}}$ edge in the line graph corresponds to the $i^{\text{th}}$ bit in $\bar{v}$ (resp. $ \bar{v}^R$). We call this process \emph{vertex expansion (VE)}. 
%For example, if $\bar{v} = 0011$, then we encode $\bar{v}$  as:
%%
%\begin{align*}
% L(v) = \blacksquare \xrightarrow[]{\text{[}}  \square \xrightarrow[]{\text{[}}  \square \xrightarrow[]{\text{(}}  \square \xrightarrow[]{\text{(}} \blacksquare
%\end{align*}
%and $\bar{v}^R = 1100$ as:
%%
%\begin{align*}
%L^R(v) =  \blacksquare \xrightarrow[]{\text{)}}  \square \xrightarrow[]{\text{)}}  \square \xrightarrow[]{\text{]}}  \square \xrightarrow[]{\text{]}} \blacksquare
%\end{align*}

\paragraph{Gadgets} Similar to the reduction in~\cite{ABW18}, we begin by constructing the set $\mC_k$ of all cliques of size $k$ in $G$. This takes time $O(n^{k})$. For each clique $t \in \mC_k$, we assume that the vertices forming $t = \{v_1, \dots, v_k\}$ are sorted in lexicographic order. We define two types of gadgets. 

\smallskip

The first gadget is the \emph{clique list} (CL). Consider a $k$-clique $t = \{v_1, v_2, \dots, v_k\}$. To create the gadget $CL(t)$, we take the line graphs $L(v_1), L(v_2), \dots, L(v_k)$ and stitch them together to form a line graph with $k \cdot N$ edges. In particular, the last node of $L(v_i)$ is the first node of $L(v_{i+1})$ for every $i=1, \dots, k-1$. 
%For each vertex $v_i \in t$, we create a directed line graph with $k^2 \log n + 1$ vertices and $k^2 \log n$ edges. Starting from vertex $v_1 \in t$, it expands into a line graph with $k \log n$ edges. The label of the $i^{\text{th}}$ edge is the $i^{\text{th}}$ bit in $\bar{v}$. We call this process as \emph{vertex expansion (VE)}. We continue vertex expansion for all vertices in $t$. The last vertex in the expansion of $v_i$ connects to the first vertex in the line graph of $v_{i+1}$. 
For simplicity of presentation, we will directly use the vertex instead of its expansion as shown in~\autoref{fig:dyck}. We also define the reverse of a clique list $\overline{CL}(t)$. Here, we take the line graphs $L^R(v_1), L^R(v_2), \dots, L^R(v_k)$ and stitch them in reverse order, i.e. the last node of $L^R(v_i)$ is the first node of $L^R(v_{i-1})$ for every $i=2, \dots, k$. 
%where the gadget construction is done on the clique vertices being sorted in reverse lexicographical order and encoding used is $\bar{v}^R$ in the construction of $\overline{CL}(t)$. 
The clique list construction does not repeat any vertex, unlike the reduction in~\cite{ABW18} that repeats each vertex $k$ times.

\smallskip

The second gadget is the \emph{clique neighbor gadget} (CNG). For a given clique $t$, consider the set of all vertices $N_t = \{w_1, \dots, w_\ell\}$ (sorted in lexicographic order), \changes{such that for every vertex $u \in N_t$ forms an edge with every vertex in $t$. Note that $t \cap N_t = \emptyset$.} These sets can be computed in time $O(n^{k+1})$ \changes{as follows. For a given clique$t$ of size $k$, we iterate over all the vertices $v \in V$ in the graph and check in constant time whether $v$ connects to all $k$ vertices of $t$. Therefore, each clique requires $O(n)$ time to be processed. As there are $n \choose k$ cliques of size exactly $k$ (these can be generated in time $O(n^k)$ straightforwardly), the total processing time is $O(n^k \cdot n) = O(n^{k+1})$.}

Then, $CNG(t)$ is a directed graph that is constructed as follows. First, for every $w_i$ we create a line graph with $k\cdot N$ edges by stitching together $k$ copies of $L(w_i)$. Then, we stitch these line graphs together by making the first node of the first copy be the same for all $w_i$, the first node of the second copy be the same, etc. The last node of the last copy is also the same (see~\autoref{fig:dyck}).
%contains $k$ copies of $VE(v)$ for each $v \in N_t$ and $k$ connector vertices that are connected as shown in~\autoref{fig:dyck}. 
% Note that we use the vertex itself instead of its expansion for simplifying the presentation. Last vertex of $VE(v)$ in the ${i}^{\text{th}}$ copy flows into the $i^{\text{th}}$ connector vertex and the connector vertex has outgoing edges to the first vertex of each ${i+1}^{\text{th}}$ $VE(v)$. The first vertex of all $VE(v)$ in the first copy and last vertex of all $VE(v)$ in the last copy are connected to the gadget connector vertices. 
 Similar to clique list, we also define $\overline{CNG}(t)$ where the gadget uses $L^R(w_i)$ for every vertex $w_i \in N_t$. 

\paragraph{Graph Construction} The instance for the \textsc{On-Demand} Dyck-$2$ language is constructed as follows. For each clique $t$, we stitch $CL_1(t)$ with $CNG_1(t)$ such that the last node of $CL_1(t)$ is the same as the first node of $CNG_1(t)$. All $CNG_1(t)$ flow into a common connector vertex $A$ (shown in black in~\autoref{fig:dyck} right after $CNG_1(t_i)$). Then, we construct $\overline{CL}_2(t)$ for each clique, make $A$ to be the start vertex of all $\overline{CL}_2(t)$, and connect it to  $CNG_2(t)$, which flows into another common connector vertex. Finally, we perform the same process but for $\overline{CL}_3(t)$ and $\overline{CNG}_3(t)$. The first vertex of every $CL_1(t)$ and last vertex of every $\overline{CNG}_3(t)$ connect to vertices $p$ and $q$ respectively. We label the outgoing edges from $p$ with $[$ and the incoming edges to $q$ with $]$.
We now use this instance as the extensional input for the \textsc{On-Demand} problem over Dyck-$2$ and ask whether $T(p,q)$ is true or not.

Before proving the result, we state a simple but useful observation that follows from the construction.

\begin{observation} \label{obv:1}
	Consider the gadget $CL(t)$ (resp. $\overline{CL}(t)$) that is immediately followed by $CNG(t)$  (resp. $\overline{CNG}(t)$). Then, the set of vertices traversed by any path in $CNG(t)$ (resp. $\overline{CNG}(t)$) has no vertex in common with $t$.
\end{observation}

\begin{claim} \label{claim:11}
	If the \textsc{On-Demand} problem on Dyck-$2$ returns true, then there exists a $3k$-clique in the input graph.
\end{claim} 

\begin{proof}
	Let $t_1$ be the clique chosen by the vertices in $CL_1(t)$ and let $V$ be the set of vertices traversed in $CNG_1(t)$ by the grammar. Since $CNG_1(t)$ is followed by $\overline{CL}_2(t)$, a valid Dyck-$2$ can be formed only if $V$ corresponds to some clique $t_2 \in \mC_k$. Indeed, if  this was not the case, then the word will not be well-formed. For instance, if $\overline{CL}_2(t)$ corresponds to the reverse of $t_1$, then we will have a set of open brackets between a set of balanced brackets which is not a valid word in Dyck-$2$. Further,~Observation~\ref{obv:1} guarantees that there is no common vertex between $t_1$and $t_2$. Thus, it holds that $t_1 \cup t_2$ is a $2k$-clique. Next, suppose $V'$ is the set of vertices traversed in ${CNG}_2(t)$. We need to argue that $V'$ corresponds to a clique $t_3$ with no common vertices with $t_1 \cup t_2$. This claim follows from the observation that since the remaining two gadgets are $\overline{CL}_3(t)$ and $\overline{CNG}_3(t)$, the Dyck-$2$ word can be valid only if the last $\overline{CNG}_3(t)$ uses the vertex set of $t_1$. Similar to our previous argument, if this was not the case, and $\overline{CL}_3(t)$ uses vertices from (say) $t_1$, the word is not balanced because $V'$ contains open brackets within the word where $t_1$ is balanced. Thus, $\overline{CNG}_3(t)$ must correspond to $t_1$. Observation~\ref{obv:1} tells us that $\overline{CL}_3(t)$ cannot have any common vertices with $\overline{CNG}_3(t_1)$. Applying Observation~\ref{obv:1} again, we see that $V'$ also cannot have any common vertex with $\overline{CL}_2(t_2)$. Thus, $t_3 \cup t_2$ and $t_3 \cup t_1$ are both $2k$-cliques. This completes the proof. 
\end{proof}

\begin{claim} \label{claim:2}
	If there exists a $3k$-clique in the input graph, then the \textsc{On-Demand} problem on Dyck-$2$ returns true.
\end{claim}

\begin{proof}
	Let $t_1, t_2, t_3 \in \mC_k$ be three disjoint $k$-cliques. We will show that there exists a path from $u$ to $v$ that forms a valid Dyck-$2$ word. Consider the path formed by the vertices 
	$$CL(t_1), CNG_2(t_2), \overline{CL}_2(t_2), CNG_2(t_3), \overline{CL}_3(t_3), \overline{CNG}_3(t_1).$$ 
The brackets on the outgoing and incoming edges from $p, q$ and all connector vertices are also balanced. It is also straightforward to see that $t_2$ and $t_3$ are balanced within the Dyck word formed by the balanced brackets of $t_1$. 
\end{proof}

Given an instance of 3k-Clique graph $G = (V,E)$, we construct the instance as described above and solve the \textsc{On-Demand} problem over the Dyck-$2$ language. By Claims~\ref{claim:11} and~\ref{claim:2}, the \textsc{On-Demand} problem returns true iff the graph $G$ contains a $3k$-clique.

\section{The All-Pairs Problem}
\label{sec:cfg}

In this section, we will study the all-pairs problem for CFL reachability. We recall here the following known facts about this problem:
\begin{itemize}
\item For every grammar $\mathcal{G}$,  $\cfga{\mathcal{G}}$ can be solved by a combinatorial algorithm in time $O(n^3)$.
\item If $L(\mathcal{G})$ is regular, then  $\cfga{\mathcal{G}}$ can be solved in $O(n^\omega)$ time via fast matrix multiplication.
\item If  $\mathcal{G}$ is a linear grammar, then $\cfga{\mathcal{G}}$ can be solved by a combinatorial algorithm in $O(m \cdot n)$ time.
\end{itemize}

\subsection{A Reduction to BMM}

It is known that for some grammar  $\mathcal{G}$, there exists a fine-grained reduction to BMM. The first question we answer is: for which CFGs can we reduce CFL reachability to BMM, in the sense that a faster running time for  $\cfga{\mathcal{G}}$ implies a faster running time for BMM?

To answer this question, we first need the following definition:

\begin{definition}[Join-Inducing CFG]
Let  $\mathcal{G}$ be a context-free grammar. We say that  $\mathcal{G}$ is {\em join-inducing} if it produces at least one string of length $\geq 2$. Otherwise, we say that  $\mathcal{G}$ is {\em join-free}.
\end{definition}

%We call a chain Datalog program join-inducing (resp. join-free) if its corresponding CFG is join-inducing (resp. join-free). 
It turns out that we can check whether this property is satisfied efficiently.

\begin{restatable}{lemma}{joininducing} \label{lem:join-inducing}
Let  $\mathcal{G}$ be a CFG. Then, in polynomial time (w.r.t. the size of $\mathcal{G}$) we can check whether  $\mathcal{G}$ is join-inducing, and if so, output a string of length $\geq 2$ produced by  $\mathcal{G}$ \changes{ in polynomial time (w.r.t. the size of $\mathcal{G}$)}.
\end{restatable}

\begin{proof}
 We say that a grammar is {\em proper} if: $(i)$ it has no rules of the form $X \gets \epsilon$ (with the exception of one rule of the form $S \gets \epsilon$ if $\mathcal{G}$ produces the empty string), $(ii)$ it has no cycles (meaning that a non-terminal symbol cannot derive itself), $(iii)$ all non-terminal symbols are productive (i.e. they can derive a word with terminal symbols), and $(iv)$ all non-terminal symbols are reachable from the start symbol $S$. We can always transform  $\mathcal{G}$ into a weakly equivalent\footnote{Weakly equivalent means that $G$ and $G'$ produce the same language.} grammar  $\mathcal{G}'$ that is proper, and we can do this in polynomial time in the size of the grammar. We can then transform  $\mathcal{G}'$ into a weakly equivalent grammar  $\mathcal{G}''$ that is in Chomsky Normal Form: this means that every production rule is of the form $A \gets B C$ or $A \gets \alpha$, where $A,B, C$ are non-terminal symbols and $\alpha$ is a terminal symbol. This can also be done in polynomial time. Note that  $\mathcal{G}''$ is also proper after this transformation.

\smallskip
\noindent {\em Claim:  $\mathcal{G}$ is join-inducing if and only if  $\mathcal{G}''$ has a rule of the form $A \gets B C$}.

\smallskip
Indeed, if every rule in  $\mathcal{G}''$ has one terminal symbol on the right hand side, then  $\mathcal{G}$ can generate only strings of length one, hence it is join-free. Otherwise, since  $\mathcal{G}''$ is proper and in Chomsky Normal Form, it must have a rule of the form $S \gets A B$. Since $A,B$ are non-terminal symbols, they must each derive a string with at least one terminal symbol. Hence, the grammar  $\mathcal{G}''$ (and thus  $\mathcal{G}$) can produce a string of at least length 2. \changes{Note that the string can be computed in polynomial time. In fact, once $\mathcal{G}''$ has been produced, generating the output takes $O(|\mathcal{G}''|)$ time: a linear pass beginning from the rule with the start symbol is sufficient to find a string of length $\geq 2$.}
\end{proof}

We can now prove the following conditional lower bound.

\begin{lemma} \label{thm:cfg}
Let  $\mathcal{G}$ be a join-inducing CFG.  Suppose that $\cfga{\mathcal{G}}$ can be computed in time $T(n) = \Omega(n^2)$. Then, BMM can be solved in time $O(T(n))$.
%Under the combinatorial BMM hypothesis, there is no combinatorial algorithm that evaluates $\cfga( \mathcal{G})$  in time $O(n^{3-\epsilon})$ for any constant $\epsilon > 0$.
\end{lemma}

\begin{proof}
By Lemma~\ref{lem:join-inducing}, we can find a string of length at least 2 that is produced by $\mathcal{G}$ (\changes{note that $|\mathcal{G}|$ is a constant}). Let this string be $r_1 r_2 \dots r_k$ where $k \geq 2$.

Now, suppose we want to multiply two $n \times n$ Boolean matrices $A,B$. We encode the matrices as a directed $(k+1)$-partite graph $H$ with vertex sets $V_0, \dots, V_{k}$ of size $n$. Let $V_\ell = \{ v_1^{(\ell)}, \dots, v_n^{(\ell)}\}$. We only add edges between two consecutive vertex sets $V_\ell, V_{\ell+1}$, where $\ell=0, \dots, k-1$ as follows:
\begin{align*}
  E_1 & = \{(v_i^{(0)}, v_j^{(1)}) \mid A[i][j] = 1, i \in [n], j \in [n] \} \\
  E_i & = \{(v_j^{(i-1)}, v_j^{(i)}) \mid j \in \{1, \dots, n \} \} \quad \text{for }  i \in \{2, \dots, k-1\} \\
  E_k & = \{(v_i^{(k-1)}, v_j^{(k)}) \mid B[i][j] = 1, i \in [n], j \in [n] \}
\end{align*}
Finally, we label the edges of $H$ such that if $(u,v) \in E_i$, then we assign label $r_i$. 

Now, consider the set the pairs $P$ produced if we run $\cfga{\mathcal{G}}$ on $H$. We take the result and filter it such that the first column has values from $V_0$ and the second from $V_k$; more specifically, we compute $P' = P \cap (V_0 \times V_k)$. Since $|V_0| \cdot |V_k| = n^2$, this computation can run in time $O(n^2)$. Observe that the input graph has $\Theta(n)$ vertices, hence the total running time is $O(T(n) + n^2) = O(T(n))$.

We now claim that $P'$ computes $C = A \times B$, i.e., $C[i][j]=1$ if and only if $(v_i^{(0)},v_j^{(k)}) \in P'$.

\smallskip

\noindent $\Rightarrow$ For the one direction, suppose that $C[i][j] = 1$. Then, there exists some $k \in \{1, \dots, n\}$ such that $A[i][k] = B[k][j] = 1$. Now, consider the following directed path in $H$:
 $$(v_i^{(0)}, v_k^{(1)}), (v_k^{(1)}, v_k^{(2)}), \dots, (v_k^{(k-1)}, v_j^{(k)})$$
First, notice that $v_i^{(0)} \in V_0$ and $v_j^{(k)} \in V_k$. Second, the word along the path is labeled $r_1 \dots r_k$, hence it is accepted by $\mathcal{G}$. From these two facts, we obtain that $(v_i^{(0)},v_j^{(k)}) \in P'$.

\smallskip

\noindent $\Leftarrow$  For the other direction, consider some $(v_i^{(0)},v_j^{(k)}) \in P'$. Since $H$ is a directed $(k+1)$-partite graph, any string that produces a result in $P'$ will be a substring of $r_1 r_2 \dots r_k$. However, the intersection with the cartesian product $V_0 \times V_k$ keeps only the strings that start from $V_0$ and end at $V_k$, so these will be exactly $r_1 r_2 \dots r_k$. Hence, there is a path $(v_i^{(0)}, v_k^{(1)}), (v_k^{(1)}, v_k^{(2)}), \dots, (v_k^{(k-1)}, v_j^{(k)})$ for some $k  \in \{1, \dots, n\}$, which means that $A[i][k] = B[k][j] = 1$ and consequently $C[i][j]=1$.
\end{proof}

The above lemma tells us that solving  $\cfga{\mathcal{G}}$ is at least as hard as BMM if $\mathcal{G}$ is join-inducing. On the other hand, the problem becomes trivial for join-free CFGs.

\begin{lemma}
Let $\mathcal{G}$ be a join-free CFG. Then, $\cfga{\mathcal{G}}$ can be evaluated in time $O(m+n)$.
\end{lemma}

\begin{proof}
Since $\mathcal{G}$ is not join-inducing, the language $L(\mathcal{G})$ can be described as a set $A$ of strings of length one (plus possibly the empty string). Hence, we can simply return the edges in the graph with labels from $A$, a task that can be done in time linear to the number of edges and nodes in the graph.
\end{proof}

\subsection{The Landscape for Combinatorial Algorithms}

Combining the results of the previous section, we can obtain the following dichotomy theorem that characterizes the complexity of the problem when we restrict to combinatorial algorithms.

\begin{theorem}\label{thm:dichotomy}
Let $\mathcal{G}$ be a context-free grammar.
\begin{itemize}
\item If $\mathcal{G}$ is join-inducing, then $\cfga{\mathcal{G}}$ can be evaluated in time $O(n^3)$ by a combinatorial algorithm. Moreover, under the combinatorial BMM hypothesis, there is no combinatorial algorithm that evaluates $\cfga{\mathcal{G}}$ in time $O(n^{3-\epsilon})$ for any constant $\epsilon > 0$.
\item If $\mathcal{G}$ is join-free, $\cfga{\mathcal{G}}$ can be evaluated in time $\Theta(m+n) = \Theta(n^2)$.
\end{itemize}
Moreover, we can decide which of the two cases holds in time polynomial to the size of $\mathcal{G}$.
\end{theorem}

The above dichotomy theorem shows a sharp behavior of the running time with respect to $n$. Indeed, the exponent of $n$ can be either 2 or 3, with nothing in between. Surprisingly, it is even decidable (in polynomial time) in which of two classes each grammar belongs.

\subsection{The Landscape for All Algorithms}

If we allow for any type of algorithm (including ones that use fast matrix multiplication), then the complexity landscape changes considerably. Indeed, we already know that regular grammars admit a subcubic algorithm. Using Lemma~\ref{thm:cfg} we can characterize the complexity within the class of CFGs that describe regular languages.

\begin{theorem}
Let $\mathcal{G}$ be a CFG such that $L(\mathcal{G})$ is regular and join-inducing. Then, $\cfga{\mathcal{G}}$ can be solved in time $O(n^\omega)$. Moreover, there is no algorithm  that evaluates $\cfga{\mathcal{G}}$ in time $O(n^{\omega-\epsilon})$ for any constant $\epsilon > 0$.
\end{theorem}

We next consider CFGs that describe non-regular languages. Is it the case that all such CFGs require cubic time even with non-combinatorial tools?  \changes{~\cite{Andersen} already showed that $\cfga{\mathcal{D}_1}$ can be solved in time $O(n^\omega \log^2 n)$. We present next another example of a natural non-regular CFG that can be solved in subcubic time. The running time for this CFG has a different exponent, which means that it possibly captures a different class of problems.}

Consider the following non-regular language: $\mathcal{L}_\geq = \{a^i b^j \mid i \geq j \}$. This language can be expressed by the following CFG $\mathcal{G}_\geq$: 
$$S \leftarrow T_1 T_2 \quad\quad T_1 \leftarrow \epsilon \mid a T_1 \quad\quad T_2 \leftarrow \epsilon \mid a T_2 b.$$

\begin{lemma}
$\cfga{\mathcal{G}_\geq}$ can be solved in time $O(n^{(3+\omega)/2})$.
\end{lemma}

\begin{proof}
The algorithm works in three steps. In the first step, we compute the all-pairs shortest paths in the graph $G_a$ where we keep only edges with label $a$ with weight $-1$, and remove any edges with label $b$. This can be done in time $O(n^{(3+\omega)/2})$~\cite{AlonGM97}. A shortest path between $i$ and $j$ in $G_a$ means a longest path of $a$-edges in $G$. Let $M_a$ be the matrix such that $M_a[i][j]$ is the length of the longest $a$-path between $i$ and $j$.

In the second step, we  compute the all-pairs shortest paths in the graph $G_b$ where we keep only edges with label $b$ with weight $+1$, and remove any edges with label $a$. A shortest path here means a shortest path of $b$-edges in $G$. This can also be done in time $O(n^{(3+\omega)/2})$~\cite{AlonGM97}. Let $M_b$ be the matrix such that $M_b[i][j]$ is the length of the shortest $b$-path between $i$ and $j$.

In the final step, we compute the existence-dominance product of the two matrices $M_a, M_b$.  The existence-dominance product of two integer matrices $A$ and $B$ is the Boolean matrix $C$ such that $C[i] [j] = 0$ iff there exists a $k$ such that $A[i][k] \geq B[k][j].$ We can solve the existence-dominance product problem also in time $O(n^{(3+\omega)/2})$~\cite{Matousek91}. 

We finally claim that the algorithm is correct. Indeed, if the resulting existence-dominance product matrix $C$ has $C[i][j]=1$, this means that there exists an $a$-path from $i$ to some $k$ that is at least as long as another $b$-path from $k$ to $j$. The concatenation of these two paths forms a path with labels thats satisfies $\mathcal{G}_\geq$.  On the other hand, if $C[i][j]=0$, for any intermediate node $k$, the longest $a$-path from $i$ to $k$ is strictly shorter that the shortest $b$-path from $k$ to $j$, hence no path from $i$ to $j$ satisfies the CFG.
\end{proof}

What is the best possible lower bound we can get for a non-regular CFG? Recall that in the previous section we showed that CFL reachability for Dyck-2 admits a $O(n^{2.5})$ conditional lower bound. This implies that solving CFL reachability for Dyck-2 (and thus in general) is likely strictly harder than BMM. However, it is an open problem if the lower bound can be improved, or if there exists a faster non-combinatorial algorithm.

\introparagraph{An Undecidability Result} We will show next that it is not possible to determine whether a given CFG can be evaluated in  time $O(n^\omega)$ or not. In other words, even if a characterization of the complexity exists for different CFGs, this characterization will not be decidable. The undecidability result is based on Greibach's theorem. First, we need the following definition.

\begin{definition}[Right Quotient] 
If $\mathcal{L}$ is a language and $\alpha$ is a single symbol, we define the language $\mathcal{L}/\alpha = \{w \mid w\alpha \in \mathcal{L} \}$ to be the {\em right quotient} of $\mathcal{L}$ with respect to $\alpha$.
\end{definition}

\begin{theorem}[Greibach~\cite{Greibach}]
Let $C$ be any non-trivial property for the class of CFLs that is true for all regular languages and that is
preserved under the right quotient with a single symbol. Then $C$ is undecidable for the class of CFLs.	 
\end{theorem}
 
 We now state the theorem formally.
 
 \begin{theorem} \label{thm:undecidable:n}
Suppose the APSP or 3SUM hypothesis holds. Then, for any constant $c \in [\omega,2.5)$, it is undecidable whether for a given CFG $\mathcal{G}$, $\cfga{\mathcal{G}}$ can be evaluated in time $O(n^c)$.
\end{theorem}
 
\begin{proof}
Fix a constant $c \in [\omega,2.5)$. To prove undecidability, we apply Greibach's theorem. Consider the following property $C$ for a CFL: the \textsc{All-Pairs} CFL reachability problem for any CFG that produces the language can be evaluated in time $O(n^c)$. As we have seen, $C$ is satisfied by all regular languages since $c \geq \omega$. It is also non-trivial, since under the APSP or 3SUM hypothesis, Dyck-$2$ cannot be evaluated in time $O(n^{2.5-\epsilon})$ for any constant $\epsilon >0$, hence it does not admit an $O(n^c)$ algorithm. It remains to show that $C$ is closed under the right quotient by a single symbol. 

\smallskip

Indeed, take a CFL $\mathcal{L}$ and a CFG $\mathcal{G}$ that produces it. Consider the language $\mathcal{L} / \alpha$ for a single symbol $\alpha$. We now want to evaluate the CFG $\mathcal{G}_\alpha$ that corresponds to the language $\mathcal{L} / \alpha$. To do this, we extend the input graph $G$ as follows: for each vertex $v \in V$, we add an edge $(v,t_v)$ with label $\alpha$, where $t_v$ is a fresh distinct vertex. Let $G'$ be the resulting graph. Note that $|V(G')| = 2 |V(G)|  = O(n)$.
Then, we run the algorithm for $\cfga{\mathcal{G}}$ on the new graph $G'$, which runs in time $ O(n^c)$. 
Finally, we can see that by construction, $(u,v)$ is an output pair for $\mathcal{G}_\alpha$ if and only if $(u,t_v)$ is an output pair for $\mathcal{G}$. Hence, to obtain the output for $\mathcal{G}_\alpha$ it remains to do the following: for every pair of the form $(u,t_v)$ for $\mathcal{G}$, output $(u,v)$. This is doable in time $O(n^2)$ by iterating over all output pairs for $\mathcal{G}$. 
\end{proof}

\subsection{All-Pairs on Sparse Graphs}

Finally, we turn our attention to the case where we interested in running time as a function of the number of edges in the graph $m$ (instead of the number of nodes). This is particularly helpful when the input graph to the CFL reachability is sparse. We summarize our results in Table~\ref{table:1}. 

As we proved in the previous subsection, a CFG that is join-free can be evaluated in time $O(m)$, which is optimal. On the other hand, we can show the following {\em unconditional lower bound} for join-inducing grammars. This result uses the fact that for any join-inducing grammar, we can construct a worst-case input instance that produces an output of size $\Omega(m^2)$. 

\begin{lemma} \label{thm:cfg2}
Let $\mathcal{G}$ be a join-inducing CFG. Then, any algorithm that computes $\cfga{\mathcal{G}}$ needs time $\Omega(m^2)$.
\end{lemma}

\begin{proof}
We follow the same construction as in the proof of \autoref{thm:cfg}. In particular, since $\mathcal{G}$ is join-inducing, it can produce a string $r_1 r_2 \dots r_k$ with $k \geq 2$. 

Consider the following family of $(k+1)$-partite graphs with vertex sets $V_0, \dots, V_{k}$ of size $n$. Let $V_i = \{ v_1^{(i)}, \dots, v_n^{(i)}\}$. We only add edges between two consecutive vertex sets $V_i, V_{i+1}$ as follows (fix some \changes{ $j^\star \in \{1, \dots, n\}$}):
\changes{
\begin{align*}
  E_1 & = \{(v_i^{(0)}, v_{j^\star}^{(1)}) \mid i \in \{1, \dots, n \} \} \\
  E_i & = \{(v_j^{(i-1)}, v_j^{(i)}) \mid j \in \{1, \dots, n \} \} \quad i \in \{2, \dots, k-1 \}\\
  E_k & = \{(v_{j^\star}^{(k-1)}, v_i^{(k)}) \mid i \in \{1, \dots, n \} \} 
\end{align*}}
Finally, we assign label $r_i$ to any edge in $E_i$.
It is easy to see that the input size is $m = k \cdot n $, while the output size is $n^2 = \Omega(m^2)$. Since any algorithm must produce this output, the desired lower bound is obtained.
\end{proof}

%\begin{lemma}
%Every linear chain Datalog program can be evaluated in time $O(m n)$.
%\end{lemma}
%
%\begin{proof}
%We decompose the chain rule as a sequence of binary joins, starting from the unique intensional atom. In this case, every rule will be of the  form $T(x_1, x_2)  \obtainedfrom T_1(x_1, y), T_2(y, x_2)$,
%where only one of $T_1, T_2$ is intensional. It is straightforward to see that each such rule has grounding of size $O(m n)$.
%\end{proof}

\def\arraystretch{1.2}
\begin{table*}[t]
	\caption{Upper and lower bounds for the \textbf{all-pairs} CFL reachability problem.}
	\label{table:1}
\centering
\begin{tabular}{|c | c|  c|} 
 \hline
  \textbf{CFG} &  \textbf{upper bound} & \textbf{lower bound} \\ 
 \hline\hline
  %$n$ &  CFG  & $O(n^3)$ \cite{Yannakakis90}  &  $\Omega(n^{3-\epsilon})$ \hfill under BMM [Thm~\ref{thm:cfg}] \\ \hline \hline
 %\multirow{2}{3em}{linear} & $O(m n)$ &   $\Omega(m n)$ (unconditional) \\ 
   join-free &  $O(m)$ & $\Omega(m)$  \hfill unconditional  \\ \hline
   join-inducing   &   $O(m^3)$ &   $\Omega(m^2)$  \hfill unconditional  [Thm~\ref{thm:cfg2}]\\ \hline
   join-inducing + linear &  $O(mn) = O(m^2)$ & $\Omega(m^2)$  \hfill unconditional  [Thm~\ref{thm:cfg2}] \\  \hline
   Dyck-1  & $O(m^3)$ & $\Omega(m^2)$  \hfill unconditional [Thm~\ref{thm:cfg2}]\\  \hline
   Dyck-$k$, $k \geq 2$  & $O(m^3)$& $\Omega(m^{3-\epsilon})\quad$ \hfill under comb. $k$-Clique [Thm~\ref{thm:dyck}]  \\
 \hline
\end{tabular}

\end{table*}

In terms of upper bounds, the problem $\cfga{\mathcal{G}}$ can always be evaluated in time $O(n^3) = O(m^3)$. Hence, every join-inducing CFG can be evaluated in time $O(m^c)$ for some exponent $c \in [2,3]$. Additionally, for linear CFGs Yannakakis~\cite{Yannakakis90}  showed that if $\mathcal{G}$ is {\em linear}, then  $\cfga{\mathcal{G}}$ can be evaluated in time $O(m n)$. Thus, we have proved the following dichotomy theorem.

\begin{theorem}
Let $\mathcal{G}$ be a linear CFG. Then:
\begin{itemize}
\item If $\mathcal{G}$ is join-inducing, then $\cfga{\mathcal{G}}$ can be evaluated in time $O(m^2)$ by a combinatorial algorithm. Moreover, every algorithm that evaluates $\cfga{\mathcal{G}}$ needs time $\Omega(m^2)$.
\item If $\mathcal{G}$ is join-free, $\cfga{\mathcal{G}}$ can be evaluated in (optimal) linear time $O(m)$.
\end{itemize}
Moreover, we can decide which of the two cases holds in time polynomial to the size of the grammar.
\end{theorem}

Unfortunately, the landscape becomes murkier for non-linear CFGs. Indeed, the $O(mn)$ algorithm is not applicable in this case, hence we do not know whether the $O(m^2)$ upper bound holds. 
As we showed in Theorem~\ref{thm:dyck}, Dyck-$k$ for $k \geq 2$ has an $\Omega(m^{3-\epsilon})$ lower bound under the combinatorial $k$-Clique hypothesis. It is an open problem whether there exists any CFGs with intermediate complexity  that can be evaluated in time $\Theta(m^c)$ for $c \in (2,3)$.

Can we determine whether a given CFG can be evaluated in $O(m^2)$ time, or in general in time $O(m^c)$ for some constant $c$ strictly smaller than 3? We answer this question negatively. 

\begin{restatable}{theorem}{undecidableone} \label{thm:undecidable}
Suppose the combinatorial $k$-Clique hypothesis holds. Then, for any constant $c \in [2,3)$, it is undecidable whether $\cfga{\mathcal{G}}$ can be evaluated by a combinatorial algorithm that runs in time $O(m^c)$.
\end{restatable}

\section{The On-Demand Problem}
\label{sec:on-demand}

%%%%%%%%%%%%%%%%%%%%%%%%%%%%%
For every CFG that corresponds to a regular grammar, we can solve the \textsc{On-Demand} problem in time $O(m) = O(n^2)$, which is optimal (we can always construct an input with size $n^2$). Hence, the CFGs of interest are the non-regular grammars. In Section~\ref{sec:dyck}, we \changes{saw} that $\cfgo{\mathcal{D}_k}$ for $k \geq 1$ is BMM-hard. We can use the same reduction to show BMM-hardness for other non-regular CFGs:

%\begin{center}
%\begin{tcolorbox}[size=small,colback=white,colframe=gray!40,coltitle=black,title=Same Generation (SG),width=6cm]
%\begin{align*}
%T(x,y) & \obtainedfrom S(x,y). \\
%T(x,y) & \obtainedfrom L(x,z_1), T(z_1, z_2), R(z_2,y).
%\end{align*}
%\end{tcolorbox}
%\end{center}
%

\begin{restatable}{theorem}{cubicn} \label{lem:dyck}
The \textsc{On-Demand} CFL reachability problem is BMM-hard for the (non-regular) CFGs that produce the following CFLs:
\begin{enumerate}
%\item The Dyck-$k$ language for any $k \geq 1$;
\item The language $\{a^i s b^i \mid i \geq 0 \}$ where $s$ can be any string, including the empty one;
\item Strings over $\{a,b\}$ where the number of $a$'s is equal to the $b$'s;
\item Palindrome strings of even (odd) length over an alphabet with at least $2$ symbols.
\end{enumerate}
\end{restatable}

The above theorem implies that the $O(n^3)$ algorithm we used for the \textsc{All-Pairs} problem is optimal for the \textsc{On-Demand} problem of all the above grammars if we restrict to combinatorial algorithms (under the combinatorial BMM hypothesis).

\def\arraystretch{1.2}
\begin{table*}[t]
\caption{Upper and lower bounds of combinatorial algorithms for the \textbf{on-demand} problem.}
	
\centering
\scalebox{0.9}{
\begin{tabular}{| c |c | c|  c|} 
 \hline
    & \textbf{CFG} &  \textbf{upper bound} & \textbf{lower bound} \\ 
 \hline\hline
  
  \multirow{4}{1em}{$n$} &\  any  & $   O(n^3)$ &  $\Omega(n^2)$ \hfill unconditional \\ \cline{2-4} 
  & regular  & $O(n^2)$ & $\Omega(n^2)$ \hfill unconditional \\ \cline{2-4}  
  &  $\{a^i b^i \mid i \geq 0 \}$ & $O(n^3)$ &  $\Omega(n^{3-\epsilon})$ \hfill under comb. BMM   [Thm~\ref{lem:dyck}]  \\ \cline{2-4}  
  & palindromes with $\geq 2$ symbols &  $O(n^3)$ & $\Omega(n^{3-\epsilon})$  \hfill under comb. BMM  [Thm~\ref{lem:dyck}]  \\ \hline  
  & Dyck-$k$, $k \geq 1$  & $O(n^3)$ & $\Omega(n^{3-\epsilon})$  \quad \hfill under comb. BMM  [Thm~\ref{thm:dyck:n}]  \\ \hline 
  \hline
 %\multirow{2}{3em}{linear} & $O(m n)$ &   $\Omega(m n)$ (unconditional) \\ 
  \multirow{4}{1em}{$m$} &   any  &  $O(m^3)$ &    $\Omega(m)$ \hfill unconditional \\ \cline{2-4}  
  & regular  & $O(m)$ &  $\Omega(m)$ \hfill unconditional \\ \cline{2-4}  
  & linear CFG &  $O(m^2)$ & ?  \\ \cline{2-4}  
  & $\{a^i b^i \mid i \geq 0 \}$ &  $O(m^2)$ & $\Omega(m^{2-\epsilon})$  \quad under comb. $k$-Clique [Thm~\ref{thm:anbn}]  \\ \cline{2-4} 
  & palindromes with $\geq 2$ symbols &  $O(m^2)$ & $\Omega(m^{2-\epsilon})$  \quad under comb. $k$-Clique [Thm~\ref{thm:anbn}]  \\ \cline{2-4} 
  & Dyck-1 & ?  & $\Omega(m^{2-\epsilon})$  \quad under comb. $k$-Clique [Thm~\ref{thm:anbn}]   \\ \cline{2-4} 
   & Dyck-$k$, $k \geq 2$  & $O(m^3)$ & $\Omega(m^{3-\epsilon})$  \quad under comb. $k$-Clique  [Thm~\ref{thm:dyck}]   \\ 
 \hline
\end{tabular}}
\label{table:2}
\end{table*}

%\begin{restatable}{lemma}{cubicn} \label{lem:dyck}
%Under the combinatorial BMM hypothesis, there is no combinatorial algorithm that evaluates the \textsc{On-Demand} problem in time $O(n^{3-\epsilon})$ for any constant $\epsilon > 0$ for the Datalog programs that produce the following CFLs:
%%
%\end{restatable}

The reader may now ask: is it true that every non-regular CFG has a cubic lower bound conditional to the combinatorial BMM hypothesis? We answer this question in the negative. Indeed, consider the following non-regular grammar $\mathcal{G}_\geq$ we defined in the previous section, that encodes the language $\{a^i b^j \mid i \geq j \}$. As we show with the next lemma, the \textsc{On-Demand} problem can be answered in time $O(m) = O(n^2)$.

\begin{lemma}\label{lem:linear:demand}
$\cfgo{\mathcal{G}_\geq}$ can be solved by a combinatorial algorithm in time $O(m)$.
\end{lemma}

\begin{proof}
Let $(s,t)$ be the input pair of constants and $G$ be the input labeled graph (with labels in $\{a,b\}$). We will do the following: for each vertex $v$, we will first compute the length $\ell_a[v]$ of the longest path from $s$ to $v$ that uses only $a$'s (we allow this length to be $+\infty$ in the case that we can make this path infinitely long, and $-\infty$ if $v$ is unreachable from $s$). We claim that we can do this in time $O(m)$. To do this, we first construct the graph $G_a$ that contains only the edges labeled with $a$, which can be done in linear time. Then, we perform a depth-first search starting from the vertex $u$ that produces as an output the strongly connected components (SCCs) of $G_a$, along with a topological sort of the SCCs. This step can also be performed in linear time. Let $C_1, C_2, \dots, C_t$ be the SCCs in the topological order. W.l.o.g., assume that $s \in C_1$. We let $\ell_a[C_1] =0$ if $|C_1| = 1$, otherwise $\ell_a[C_1] = +\infty$. Then, we iterate over all SCCs following the topological order: for $C_i$, if $|C_i| > 1$ we assign $\ell_a[C_i] = +\infty$; otherwise, $\ell_a[C_i] = \max_{j} \ell_a[C_j] + 1$, where $j$ iterates over all SCCs $C_j$ such that there is an edge from $C_j$ to $C_i$ (if there is no such edge, we let $\ell_a[C_i] = - \infty$). Finally, we let $\ell_a[v] = \ell_a[C_i]$, where $v \in C_i$.

\smallskip

Similarly, for each vertex $v$ we compute the length $\ell_b[v]$ of the shortest path from $v$ to $t$ that uses only $b's$ (the distance can be $+\infty$ if there is no such path). We can also do this in time $O(m)$. Indeed, we first construct the graph $G_b$ that contains only the edges labeled with $b$. Then, we solve a single-target shortest path problem in an unweighted graph, where the target is $t$; this can be done in linear time using breadth-first search.

\smallskip

Finally, we iterate over every node $v \in V(G) \setminus \{s,t\}$. If there exists such a node where $\ell_a[v] \neq - \infty$, $\ell_b[v] \neq + \infty$ and $\ell_a[v] > \ell_b[v]$ then we claim that the desired path between $s,t$ exists; otherwise not. The final step can be done in time $O(n)$.
\end{proof}

Finally, we can show an analogous undecidability result to the one in Theorem~\ref{thm:undecidable}.

\begin{restatable}{theorem}{undecidable}\label{thm:undecidable2}
Suppose the combinatorial BMM hypothesis holds. Then, for any constant $c \in [2,3)$, it is undecidable whether the \textsc{On-Demand} CFL reachability problem for a given CFG can be evaluated by an $O(n^c)$  combinatorial algorithm.
\end{restatable}

%%%%%%%%%%%%%%%%%%%%%%%%%%%%%
\introparagraph{On-demand on Sparse Graphs} Finally, we study the \textsc{On-Demand} problem with respect to the input size $m$. We identify CFGs that can be evaluated in linear, quadratic, and cubic time (see Table~\ref{table:2}).  The \textsc{On-Demand} problem  for a program that corresponds to a regular CFG is in time $O(m)$~\cite{Yannakakis90}, while there are non-regular programs that are also in linear time. On the other hand, there are several non-regular (and even linear)  programs that have a quadratic lower bound. To show this  lower bound, we use the fact that, under the combinatorial $k$-Clique hypothesis, for any constant $\epsilon >0$, for any $k> 2/\epsilon$, finding whether a graph contains a $k$-cycle can be detected in time $\Omega(m^{2-\epsilon})$~\cite{Lincoln020}.

\begin{restatable}{theorem}{anbn} \label{thm:anbn}
Under the combinatorial $k$-Clique hypothesis, the \textsc{On-Demand} problem for following cases cannot be solved by a combinatorial algorithm in time $O(m^{2-\epsilon})$ for any constant $\epsilon > 0$.
\begin{enumerate}
\item Dyck-$k$, for any $k \geq 1$;
\item The language $\{a^i s b^i \mid i \geq 0 \}$ where $s$ can be any string, including the empty one;
\item Strings over $\{a,b\}$ where the number of $a$'s is equal to the $b$'s;
\item Palindrome strings of even (odd) length over an alphabet with at least $2$ symbols.
\end{enumerate}
\end{restatable}

We can prove lower bounds for other CFGs using {\em inverse homomorphisms}. A homomorphism $h$ on an alphabet $\Sigma$ is a function that gives a word (in a possibly different alphabet) for every symbol in $\Sigma$. We can extend $h$ naturally to map a word to another word. If $h$ is a homomorphism and $L$ a language whose alphabet
is the output language of $h$, then we define the inverse homomorphism as $h^{-1}(L) = \{ w \mid h(w) \in L\}$. CFLs are known to be closed under inverse homomorphisms.

\begin{restatable}{lemma}{inversehomomorphism} \label{lem:inverse}
Suppose the CFG that corresponds to a CFL $L$ admits an $O(m^c)$ algorithm for the \textsc{All-Pairs} (resp. \textsc{On-Demand}) problem for some constant $c \geq 1$. Then, the CFG that corresponds to $ h^{-1}(L)$ also admits an $O(m^c)$ algorithm for the \textsc{All-Pairs} (resp. \textsc{On-Demand}) problem.
\end{restatable}

\begin{example}
To show how to apply Lemma~\ref{lem:inverse} to obtain further lower bounds, consider the 
%following Datalog program $\mathsf{SG}_1$:
%\begin{align*}
%T(x,y) & \obtainedfrom S(x,y). \\
%T(x,y) & \obtainedfrom L_1(x,w_1), L_2(w_1,z_1), T(z_1, z_2), R(z_2,y).
%\end{align*}
%This program corresponds to the 
CFL $L_1 = \{ (ad)^i c b^i \mid i \geq 0\}$. Now, take the CFL $L_0 =  \{ a^i c b^i \mid i \geq 0 \}$. Consider the homomorphism $h$ with $h(a) = ad, h(b)=b$, and $h(c) = c$. It is easy to see that $L_0 = h^{-1}(L_1)$. Lemma~\ref{lem:inverse} now tells us that the $\Omega(m^{2-\epsilon})$ lower bound for the \textsc{On-Demand} version of $L_1$ holds for $L_0$ as well.
\end{example}

Some CFGs, even under sparse inputs, do not admit a truly subcubic combinatorial algorithm. Abboud et. al~\cite{ABW18} constructed a (fairly complex) CFG for which parsing is not truly subcubic. Since parsing corresponds to running CFL reachability over a graph that is a path (and hence $m = n-1$), this construction already finds a grammar with the desired lower bound. But as we showed in Section~\ref{sec:dyck}, Dyck-$2$ already achieves this lower bound.

We end this section with another undecidability result, which shows that we cannot even hope to determine for which programs the \textsc{On-Demand} problem can be evaluated in linear time. 

\begin{restatable}{theorem}{undecidablethree}\label{thm:undecidable3}
Suppose the combinatorial $k$-Clique hypothesis holds. Then, for any constant $c \in [1,3)$, it is undecidable whether the \textsc{On-Demand} problem for a CFG  can be evaluated by an $O(m^c)$  combinatorial algorithm.
\end{restatable}

\section{A Lower Bound for Andersen's Pointer Analysis}
\label{sec:binary}

\begin{figure}[!t]
	\scalebox{0.7}{
	\begin{tikzpicture}
		\def\cl1{1}
		\def\cng1{2.2}
		\def\cngfig{3.85}
			
			\tikzstyle{every node}=[font=\Large]

			%% first CL1 gadget %%%%%%%%%%%%%%%%
			\node (v1)  at (-8.5, 3.75) { $\blacksquare$};
			\node (v2)  at (-7.5, 3.75){ $\blacksquare$};
			\node (d1)  at (-6.5, 3.75) { $\dots$};
			\node (vk)  at (-5.5, 3.75) { $\blacksquare$};
			
			\draw[->] (v1) -- (v2) node[midway,above]{$v_1$} ;
			\draw[->] (v2) -- (d1) node[midway,above]{$v_2$} ;
			\draw[->] (d1) -- (vk);
			
			\draw [decorate,decoration = {brace, mirror}] (-8.5,3.5) --  (-4.6,3.5);
			\node at (-6.5, 3.2) {\small $CL_1(t_1)$};
			%% first CL1 gadget %%%%%%%%%%%%%%%%
			
			\node at (-6.5, 1.75) { $\vdots$};
			
			%% last CL1 gadget %%%%%%%%%%%%%%%%
			\node (v1k)  at (-8.5, 0.9-\cl1) { $\blacksquare$};
			\node (v2k)  at (-7.5, 0.9-\cl1) { $\blacksquare$};
			\node (d1k)  at (-6.5, 0.9-\cl1) { $\dots$};
			\node (vkk)  at (-5.5, 0.9-\cl1) { $\blacksquare$};
			
			\draw[->] (v1k) -- (v2k) node[midway,above]{$u_1$} ;
			\draw[->] (v2k) -- (d1k) node[midway,above]{$u_2$} ;
			\draw[->] (d1k) -- (vkk);
			
			\draw [decorate,decoration = {brace, mirror}] (-8.5,0.65-\cl1) --  (-4.6,0.65-\cl1);
			\node at (-6.5, 0.3-\cl1) {\small $CL_1(t_{|\mC_k|})$};
			%% last CL1 gadget %%%%%%%%%%%%%%%%
			
			\draw [decorate,decoration = {brace, mirror}] (-4.5,2.5) --  (-0.5,2.5);
			\node at (-2.35, 2.15) {\small $CNG_1(t_1)$};
			
			\draw [decorate,
			decoration = {brace, mirror}] (-4.5,0.75-\cng1) --  (-0.5,0.75-\cng1);
			\node at (-2.35, 0.35-\cng1) {\small $CNG_1(t_{|\mC_k|})$};
			
			\node (p) at (-9.5, 1.75) { $p$};
			\draw[->] (p) |- (v1)
			node[midway,below, xshift = 4pt, yshift = -5pt]{$\alpha$} ;
			\draw[->] (p) |- (v1k)
			node[midway,above, xshift = 4pt, yshift = 10pt]{$\alpha$} ;
			
			%\draw [decorate,decoration = {brace, mirror}] (-9, 5) --  (-9, -1.5);
			%\draw [decorate,decoration = {brace}] (0, 5) --  (0, -1.5);

			\node (center1) at ( -4.5, 3.75) {$\blacksquare$};
			
			\draw[->] (vk) -- (center1) node[midway,above]{$v_k$} ;
			
			\node (x1)  at (-4, 5) { $x_1$};
			\node (x2)  at (-4, 4.5) { $x_2$};
			\node (xp)  at (-4, 2.75) { $x_\ell$};
			\node (sp1) at (-3.5, 3.75) {$\blacksquare$};
			\node (sp2) at (-2.5, 3.75) {$\blacksquare$};
			\node (sp3) at (-1.5, 3.75) {$\blacksquare$};
			
			\draw[->] (center1) |- (x1);
			\draw[->] (center1) |- (x2);
			\draw[->] (center1) |- (xp);
			\draw[->] (x1) -| ($(sp1.north)-(2pt,0)$);
			\draw[->] (x2) -| ($(sp1.north)-(2pt,0)$);
			\draw[->] (xp) -| ($(sp1.south)-(2pt,0)$);

			\node (x11)  at (-3, 5) { $x_1$};
			\node (x111) at (-2, 5) { $\dots$};
			\node  (x1111) at (-1, 5) { $x_1$};
			
			\node (x21)  at (-3, 4.5) { $x_2$};
			\node (x211) at (-2, 4.5) { $\dots$};
			\node  (x2111) at (-1, 4.5) { $x_2$};
			
			\node (xp1)  at (-3, 2.75) { $x_\ell$};
			\node (xp11) at (-2, 2.75) { $\dots$};
			\node  (xp111) at (-1, 2.75) { $x_\ell$};
			
			\node (x3)  at (-4, 3.75) {$\vdots$};
			\node (x31)  at (-3, 3.75) {$\vdots$};
			\node (d4)  at (-2, 3.75) { $\vdots$};
			\node (x311)  at (-1, 3.75) {$\vdots$};
			
			\draw[->] ($(sp1.north)+(2pt,0)$) |- (x11);
			\draw[->] ($(sp1.north)+(2pt,0)$) |- (x21);
			\draw[->] ($(sp1.south)+(2pt,0)$) |- (xp1);
			\draw[->] (x11) -| (sp2);
			\draw[->] (x21) -| (sp2);
			\draw[->] (xp1) -| (sp2);
			\draw[->] (sp3) |- (x1111);
			\draw[->] (sp3) |- (x2111);
			\draw[->] (sp3) |- (xp111);
			
			\node[circle, draw, minimum size=10pt,inner sep=1pt, outer sep=2pt] (center2) at ( 0.25, 1.75) {$A$};
			
			\draw[->] (x1111) -| ($(center2.north)-(3pt,0)$);		
			\draw[->] (x2111) -| ($(center2.north)-(3pt,0)$);	
			\draw[->] (xp111) -| ($(center2.north)-(3pt,0)$);	
			
			\node (center1k) at ( -4.5, 3.75-\cngfig) {$\blacksquare$};
			
			\draw[->] (vkk) -- (center1k) node[midway,above]{$u_k$} ;
			
			\node (x1k)  at (-4, 5-\cngfig) { $x'_1$};
			\node (x2k)  at (-4, 4.5-\cngfig) { $x'_2$};
			\node (d4k)  at (-2, 3.75-\cngfig) { $\vdots$};
			\node (xpk)  at (-4, 2.75-\cngfig) { $x'_\ell$};
			\node (sp1k) at (-3.5, 3.75-\cngfig) {$\blacksquare$};
			\node (sp2k) at (-2.5, 3.75-\cngfig) {$\blacksquare$};
			\node (sp3k) at (-1.5, 3.75-\cngfig) {$\blacksquare$};
			\node  at (-4, 3.75-\cngfig) { $\vdots$};
			\node  at (-3, 3.75-\cngfig) { $\vdots$};
			\node  at (-1, 3.75-\cngfig) { $\vdots$};
			
			\draw[->] (center1k) |- (x1k);
			\draw[->] (center1k) |- (x2k);
			\draw[->] (center1k) |- (xpk);
			\draw[->] (x1k) -| ($(sp1k.north)-(2pt,0)$);
			\draw[->] (x2k) -| ($(sp1k.north)-(2pt,0)$);
			\draw[->] (xpk) -| ($(sp1k.south)-(2pt,0)$);

			\node (x11k)  at (-3, 5-\cngfig) { $x'_1$};
			\node (x111k) at (-2, 5-\cngfig) { $\dots$};
			\node  (x1111k) at (-1, 5-\cngfig) { $x'_1$};

			\node (x21k)  at (-3, 4.5-\cngfig) { $x'_2$};
			\node (x211k) at (-2, 4.5-\cngfig) { $\dots$};
			\node  (x2111k) at (-1, 4.5-\cngfig) { $x'_2$};
			
			\node (xp1k)  at (-3, 2.75-\cngfig) { $x'_\ell$};
			\node (xp11k) at (-2, 2.75-\cngfig) { $\dots$};
			\node  (xp111k) at (-1, 2.75-\cngfig) { $x'_\ell$};
			
			\draw[->] ($(sp1k.north)+(2pt,0)$) |- (x11k);
			\draw[->] ($(sp1k.north)+(2pt,0)$) |- (x21k);
			\draw[->] ($(sp1k.south)+(2pt,0)$) |- (xp1k);
			\draw[->] (x11k) -| (sp2k);
			\draw[->] (x21k) -| (sp2k);
			\draw[->] (xp1k) -| (sp2k);
			\draw[->] (sp3k) |- (x1111k);
			\draw[->] (sp3k) |- (x2111k);
			\draw[->] (sp3k) |- (xp111k);
			\draw[->] (x1111k) -| ($(center2.south)-(3pt,0)$);		\draw[->] (x2111k) -| ($(center2.south)-(3pt,0)$); \draw[->] (xp111k) -| ($(center2.south)-(3pt,0)$);

			\def\off{1}
			\node[circle, draw, minimum size=10pt, inner sep=1pt, outer sep=2pt] (center5) at ( 5.25, 1.75) {$B$};
			
			%% Second bracket %%%%%%%%%%%%%%%%%%%%%%%%%%%			
			\node (v12)  at (1, 3.75) {$\blacksquare$};
			\node (d12)  at (1+\off, 3.75) { \small $\dots$};
			\node (center4) at ( 1+\off+\off, 3.75) {$\blacksquare$};
			\node (center'4) at ( 1+\off+\off+0.75, 3.75) {$\blacksquare$};
			
			\draw[->] (v12) -- (d12) node[midway,above]{$v_{k}^R$} ;
			\draw[->] (d12) -- (center4) node[midway,above]{$v_1^R$} ;
			\draw[->] ($(center2.north)+(3pt,0)$)	 |- (v12) node[midway,below, xshift = 4pt, yshift = -5pt]{$\alpha$} ;
			\draw[->] (center4) -- (center'4) node[midway,above]{$\gamma$} ;
			
			\draw [decorate,decoration = {brace, mirror}] (1,3.5) --  (2.95,3.5);
			\node at (2.0, 3.2) {\small $\overline{CL}_2(t_1)$};
			
			\draw [decorate,decoration = {brace, mirror}] (3.05+0.65,3.5) --  (5,3.5);
			\node at (4+0.4, 3.2) {\small ${CNG}_2(t_1)$};
			\node   at (4+0.35, 3.75) { $\dots$};
			\node   at (2.75, 2) { $\vdots$};

			\draw[->] (4.75, 3.75) -| ($(center5.north)-(3pt,0)$);
			\draw[->] (4.75, 3.75 - \cngfig) -| ($(center5.south)-(3pt,0)$);

			\node (v12k)  at (1, 3.75 - \cngfig) {$\blacksquare$};
			\node (d12k)  at (1+\off, 3.75- \cngfig) { \small $\dots$};
			%\node (vk2k)  at (2.5, 3.75- \cngfig) { \small $v_1$};
			\node(center4k) at ( 1+\off+\off, 3.75- \cngfig) {$\blacksquare$};
			\node (center'4k) at ( 1+\off+\off+0.65, 3.75- \cngfig) {$\blacksquare$};
			\draw[->] (center4k) -- (center'4k) node[midway,above]{$\gamma$} ;
			
			\draw[->] (v12k) -- (d12k) node[midway,above]{$u_{k}^R$} ;
			\draw[->] (d12k) -- (center4k) node[midway,above]{$u_1^R$} ;
			\draw[->] ($(center2.south)+(3pt,0)$) |- (v12k)
			node[midway,above, xshift = 4pt, yshift = 10pt]{$\alpha$} ;
			
			\draw [decorate,
			decoration = {brace, mirror}] (1,3.5- \cngfig) --  (2.95,3.5- \cngfig);
			\node at (2.0, 3.2- \cngfig) {\small $\overline{CL}_2(t_{|\mC_k|})$};
			\draw [decorate,
			decoration = {brace, mirror}] (3.05+0.5,3.5- \cngfig) --  (5,3.5- \cngfig);
			\node at (4+0.35, 3.2- \cngfig) {\small ${CNG}_2(t_{|\mC_k|})$};
			\node   at (4+0.35, 3.75- \cngfig) { $\dots$};

			%%%%content inside third brace
			\node (v13)  at (6, 3.75) {$\square$};
			\node (d13)  at (6+\off, 3.75) { \small $\dots$};
			%\node (vk3)  at (7.5, 3.75){$\square$};
			\node (center6) at (6+\off+\off, 3.75) {$\blacksquare$};
			\node (center'6) at ( 8+1, 3.75) {$\blacksquare$};
			\draw[->] (center6) -- (center'6) node[midway,above]{\small $\bar{\gamma}\bar{\alpha}$} ;

			\draw[->] (v13) -- (d13) node[midway,above]{$v_{k}^R$} ;
			\draw[->] (d13) -- (center6) node[midway,above]{$v_1^R$} ;
			\draw[->] ($(center5.north)+(3pt,0)$) |- (v13) 
			node[midway,below, xshift = 4pt, yshift = -5pt]{$\alpha$} ;
			
			\draw [decorate,
			decoration = {brace, mirror}] (6,3.5) --  (7.95,3.5);
			\node at (7.0, 3.2) {\small $\overline{CL}_3(t_1)$};
			\draw [decorate,
			decoration = {brace, mirror}] (8.05+1,3.5) --  (10,3.5);
			\node at (9+0.6, 3.2) {\small $\overline{CNG}_3(t_1)$};
			\node   at (9+0.5, 3.75) { $\dots$};
			\node   at (7.75, 2) { $\vdots$};
			
			\node (center7) at ( 10.5, 1.75) {$q$};
			
			\draw[->] (9.95, 3.75) -| (center7)
			node[midway,below, xshift = 4pt, yshift = -5pt]{$\beta$} ;
			\draw[->] (9.95, 3.75 - \cngfig) -| (center7)
			node[midway,above, xshift = 4pt, yshift = 10pt]{$\beta$} ;

			\node (v13k)  at (6, 3.75  - \cngfig) {$\square$};
			\node (d13k)  at (6+\off, 3.75- \cngfig) { \small $\dots$};
			\node(center6k) at ( 8, 3.75- \cngfig) {$\blacksquare$};
			\node (center'6k) at ( 8+1, 3.75- \cngfig) {$\blacksquare$};
			\draw[->] (center6k) -- (center'6k) node[midway,above]{\small $\bar{\gamma}\bar{\alpha}$} ;
			
			\draw[->] (v13k) -- (d13k) node[midway,above]{$u_{k}^R$} ;
			\draw[->] (d13k) -- (center6k) node[midway,above]{$u_1^R$} ;
			%\draw[->] (vk3k) -- (center6k);
			\draw[->] ($(center5.south)+(3pt,0)$) |- (v13k)
			node[midway,above, xshift = 4pt, yshift = 10pt]{$\alpha$} ;
			
			\draw [decorate,
			decoration = {brace, mirror}] (6,3.5- \cngfig) --  (7.95,3.5- \cngfig);
			\node at (7.0, 3.2- \cngfig) {\small $\overline{CL}_3(t_{|\mC_k|})$};
			\draw [decorate,
			decoration = {brace, mirror}] (8.05+1,3.5- \cngfig) --  (10,3.5- \cngfig);
			\node at (9.1+0.5, 3.175- \cngfig) {\small $\overline{CNG}_3(t_{|\mC_k|})$};
			\node   at (9+0.5, 3.75- \cngfig) { $\dots$};

	\end{tikzpicture}}
	\caption{Input graph constructed for Andersen's analysis.} \label{fig:andersen}
\end{figure}
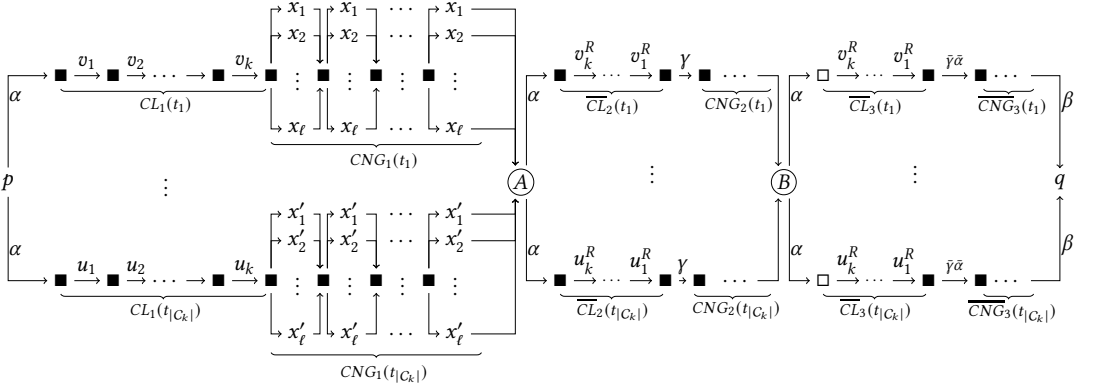
 
In this section, we show a conditional lower bound for the \textsc{On-Demand} problem for Andersen's Pointer analysis (APA). Andersen's analysis is a fundamental points-to analysis that produces an over-approximation of the memory locations that each pointer may point-to at runtime. It will be convenient for our purposes to define APA as a Datalog program that computes the inverse points-to relation $T(x,y)$: this means that variable $y$ points to variable $x$.
\begin{align}
 T(x,y) &\obtainedfrom \alpha(x,y). \\
 T(x,y) &\obtainedfrom T(x,z), e(z,y). \\
 T(w,y) & \obtainedfrom  T(w,z), T(z,x), \beta(x,y). \\
 T(w,z) & \obtainedfrom T(w, x), \gamma(x, y), T(z,y).
\end{align}
This is not a chain Datalog program, and hence APA cannot be expressed as a CFL reachability problem directly. However, following a technique from~\cite{Reps98}, we can define $\bar T$ as the inverse of $T$, i.e. $\bar T(x,y) = T(y,x)$.  With this, we can think of the  program above as the following grammar:
$$ T \gets \alpha \mid T e \mid  T T \beta \mid T \gamma \bar T.$$

In recent work~\cite{Andersen}, it was shown that the \textsc{On-Demand} problem for APA  has a $O(n^3)$ combinatorial lower bound under the combinatorial $k$-Clique hypothesis. We strengthen this result by showing that the same cubic lower bound holds even on sparse graphs.

\begin{theorem} \label{thm:andersen}
Under the combinatorial $k$-Clique hypothesis, the \textsc{On-Demand} problem for APA cannot be solved by a combinatorial algorithm in time $O(m^{3-\epsilon})$ for any constant $\epsilon > 0$. 
\end{theorem}

In the remaining section, we provide a proof of the above theorem. The proof uses the same construction as the lower bound for Dyck-2. In particular, we will start with a graph $G$ and attempt to find a $3k$-clique. The only difference is how the numbers are encoded. 
For the lower bound we will need only rules 1,3 and 4;  hence we can ignore the label $e$ in the constructed instance and we will ignore rule 2 completely. In other words, it suffices to consider the grammar $ T \gets \alpha \mid  T T \beta \mid T \gamma \bar T$.
It will be convenient to think that the labels of a path from $u$ to $v$ that is recognized by the program form a word with forward and backward edges. We will write a backward edge with label $\ell$ as a forward edge labeled $\bar \ell$.  For example, if we have a path of the form $\alpha(u,w), \gamma(w,z), \alpha(v,z)$, we will think of it as $\alpha \gamma \bar{\alpha}$. We will also use the notation $\ell^i$ to denote $i$ repetitions of the label $\ell$.

\paragraph{Notation} We associate with each vertex an integer in $\{1, \dots, n \}$. As with the construction for Dyck-2, we will create two line graphs for a vertex $v$, $L(v)$ and $L^R(v)$.
The line graph $L(v)$ has $v$ edges labeled $\alpha$ followed by two edges with labels $\alpha \gamma$ (so it forms the string $\alpha^v \alpha \gamma$).  The line graph $L^R(v)$ has two edges with labels $\bar \gamma \bar \alpha$ followed by $v$ edges labeled $\beta$  (so it forms the string  $ \bar \gamma \bar \alpha \beta^v$).  

\paragraph{Graph Construction} The construction follows the one for Dyck-2, with the only difference that we need a few additional edges as shown in Figure~\ref{fig:andersen}.  We now ask whether $T(p,q)$ is true or not. 

\paragraph{Correctness} We show in Appendix~\ref{sec:apa} that the On-Demand problem on APA returns true if and only if there exists a $3k$-Clique in the constructed graph. The resulting graph has $O(n^{k+2})$ edges and can be constructed in the same amount of time. To obtain the desired bound, we then let $k$ grow depending on the constant $\epsilon$.

\section{Related Work}
\label{sec:related}

\introparagraph{Static Program Analysis} The connection between CFL reachability and program analysis has been observed since a long time~\cite{Reps98,DBLP:conf/datalog/SmaragdakisB10, whaley2005using, melski2000interconvertibility}. \changes{Cubic time complexity is a common feature of algorithms proposed in several of these works. Prior work~\cite{heintze1997cubic} has also explained the sub-cubic barrier by showing that several data-flow reachability problems are 2NPDA-hard, a complexity class that does not admit sub-cubic algorithm for problems lying in this class. The related problem of \emph{certifying} whether an instance of CFL reachability has a small and efficiently checkable certificate was studied in~\cite{chistikov2022subcubic}. Their main result shows that succint certificates of size $O(n^2)$ can be checked in sub-cubic time using matrix multiplication.

} Several variants of Andersen's analysis~\cite{andersen1994program} have been developed over the years that incorporate different features such as flow and field-sensitivity~\cite{hirzel2004pointer, whaley2002efficient, pearce2004online, lyde2015control}. However, the study of the precise complexity is a relatively recent effort. In this direction, the authors in~\cite{Andersen} explored the fine-grained complexity of Andersen's analysis. The complexity of the Dyck Reachability problem, which can be captured as a Context-Free grammar, has also been studied previously~\cite{CCP18}. Many program analysis tasks can also be expressed by Interleaved Dyck, which is the intersection of multiple Dyck languages based on an interleaving operator. Several recent works have also established the precise combined complexity~\cite{li2020fast, li2021complexity} and fine-grained complexity~\cite{kjelstrom2021decidability} of this problem. 
%Synthesis of chain Datalog programs for many program analysis tasks~\cite{smaragdakis2010using, si2018syntax} has been found to be beneficial in practice.
\changes{Fine-grained complexity and parameterized complexity based lower bounds have also found great success for other related problems such as finding violations in concurrent programs and checking program consistency~\cite{chini2017complexity, chini2020framework} and safety verification~\cite{chini2020fine, chini2018fine}. We refer the reader to~\cite{chini2022fine} for more details on application of fine-grained complexity for program verification.}
\begin{sloppypar}
\introparagraph{Datalog} CFL reachability is essentially equivalent to a class of Datalog programs called {\em chain Datalog programs}. The seminal work of Yannakakis~\cite{Yannakakis90}  established a tight upper bound for evaluation of chain Datalog programs for regular and linear languages. This raises the question of when we can rewrite a non-linear Datalog program into a linear program, which has been studied extensively~\cite{afrati1997chain, afrati2003linearisability, dong1992datalog, afrati1996transformations, ullman1988parallel}. When we restrict CFL reachability to regular languages, then the problem is equivalent to the evaluation of  Regular Path Queries (RPQs)~\cite{Graph}. In particular,~\cite{martens2018evaluation} studied the parameterized complexity of RPQ evaluation over graphs. Bagan et al.~\cite{bagan2013trichotomy} characterized the class of regular languages that are tractable for RPQs.
%made the fine-grained distinction of RPQ evaluation over simple paths vs arbitrary paths and studied the data complexity for several extensions of RPQs that allow conjunctions and unions. 
More recently,~\cite{RPQs} studied the fine-grained static and dynamic complexity of RPQ evaluation, enumeration, and counting problems. This is similar to our effort in this paper with the difference that we study the data complexity of a fixed RPQ. 
\end{sloppypar}
%Variants of Datalog that incorporate disjunctions~\cite{eiter1994adding}, integrity constraints~\cite{cali2009datalog}, and monadic Datalog~\cite{gottlob2004monadic} have also been studied extensively. 
%~\cite{deutch2014circuits, deutch2015selective} studied the complexity of provenance representation for linear Datalog programs.
\begin{sloppypar}
\introparagraph{Fine-grained Complexity} The area of fine-grained complexity attempts to prove the optimality of several well-known algorithms by constructing reductions to problems with widely believed lower bound conjectures. Such problems are BMM, the 3-SUM problem (and $k$-SUM more generally), all-pairs shortest paths, cycle detection, and finding orthogonal vectors~\cite{williams2018some}. This line of work has conditionally proved cubic or quadratic lower bounds for sparse and dense variants of problems such as CFG parsing~\cite{ABW18}, finding subgraphs in graphs~\cite{WX20}, variations of path-related problems~\cite{DBLP:conf/soda/LincolnWW18}, and dynamic problems~\cite{DBLP:conf/focs/AbboudW14}. The development of these conjectures has led to widespread activity in establishing lower bounds for problems such as join query processing~\cite{berkholz2017answering, keppeler2020answering, carmeli2021enumeration, berkholz2017answeringucq}, concurrency analysis~\cite{kulkarni2021dynamic, mathur2020complexity}, and cryptography~\cite{lavigne2019public, golovnev2020data} to name a few.
\end{sloppypar}
\section{Conclusion}

In this work, we take the first step towards studying the fine-grained complexity landscape for the CFL reachability problem. We identify the precise polynomial running time (under widely believed lower bound conjectures) for several fundamental grammars. Despite the significant progress we made, there are many exciting questions that we have left open.

%\introparagraph{Bounds for non-combinatorial algorithms} In this paper, we focused mostly on combinatorial algorithms. However, in many cases we can obtain faster algorithms if we use algebraic techniques. For example, a chain Datalog program that corresponds to a regular grammar can be evaluated in time $O(n^\omega)$~\cite{RPQs} using fast matrix multiplication, where $\omega < 3 $ is the matrix multiplication exponent. Our lower bound constructions transfer to non-combinatorial algorithms as well, but it remains open to show for which programs we can construct faster non-combinatorial algorithms.

\introparagraph{The running time for Dyck-$1$} \changes{Prior work} has established that Dyck-1 has a combinatorial cubic lower bound w.r.t. $n$, but its running time w.r.t. $m$ remains open. We were not able to find whether the running time is cubic, quadratic, or somewhere in between. In general, we do not even know whether there exists a grammar with running time $O(m^c)$ for some constant $c \in (2,3)$.

\introparagraph{The lower bound for Dyck-2} Even though we have established the complexity of Dyck-2 in the combinatorial setting, the general running time remains open. It is still possible that a cubic lower bound exists, but it is also an intriguing possibility that a truly sub-cubic algorithm that uses fast matrix multiplication exists.

\introparagraph{Faster Algorithms for Restricted Inputs} Our algorithmic techniques are designed to work for worst-case inputs. However, restricting the instances we consider (e.g., instances with bounded treewidth) can potentially lead to faster algorithms. It may also possible to obtain better bounds if we take the output size into account.

%\introparagraph{Fine-grained complexity for CQs} The fine-grained complexity landscape is unknown even for single-rule non-recursive Datalog programs, which correspond to CQs. In Appendix~\ref{sec:loomis}, we prove that the $O(m^{k/(k-1)})$ upper bound is optimal for Loomis-Whitney joins with $k$ variables under the combinatorial $k$-Clique hypothesis. Lower bounds for 1-series parallel graph queries have also been shown conditional to the 3-SUM conjecture~\cite{DBLP:journals/mst/JoglekarR18}. We believe that similar conditional lower bounds can be obtained for other CQs.

%% Acknowledgments
\begin{acks}  
  This research was supported by
  \grantsponsor{GS100000001}{National Science
    Foundation}{http://dx.doi.org/10.13039/100000001} under Grant
  No.~\grantnum{GS100000001}{IIS-1910014}.  Any opinions, findings, and
  conclusions or recommendations expressed in this material are those
  of the authors and do not necessarily reflect the views of the
  National Science Foundation.
\end{acks}

\appendix

\section{Error in the published version} \label{sec:error}

First, we reproduce the erroneous proof verbatim as it appeared in the published version of the paper (Theorem $3.2$ in the published version).

\begin{theorem}\label{thm:bound}
Under the APSP or 3SUM hypothesis, $\cfga{\mathcal{D}_k}$ for $k \geq 2$ does not admit an $O(n^{2.5-\epsilon})$ algorithm for any constant $\epsilon >0$.
\end{theorem}

\begin{proof}
We show a reduction from the All-Edges Monochromatic Triangle problem (AE-Mono$\Delta$). In this problem, we are given an $n$-node graph $H = (V, E)$, where each edge $e \in E$ has a color $c(e)$. We ask to determine for every edge $e$, whether it appears in a monochromatic triangle in $H$, i.e., all edges of the triangle have the same color. We will use the fact that AE-Mono$\Delta$ does not admit an $O(n^{2.5-\epsilon})$ algorithm unless both the APSP and 3SUM hypotheses fail~\cite{WX20}.

\smallskip

The key idea in the reduction is that Dyck-2 (and hence Dyck-$k$ for $k \geq 2$) can encode numbers. We associate with each edge in the graph an integer in $\{1, \dots, m \}$, and with each color an integer $m+1, \dots, m+C$, where $C$ is the number of distinct colors. Let $\bar{e}$ denote the binary encoding of an edge $e$, and $\bar{e}^R$ denote the reverse sequence of $\bar{e}$. Similarly, we use $\bar{c}$ and $\bar{c}^R$ for the binary encoding of colors. We will assume that the length of the binary encoding is exactly $N = \log (m+C)$ (we can always pad with 0's). Since we will construct a Dyck-$2$ instance, bits $0, 1$ in $\bar{e}, \bar{c}$ will be encoded using symbols $[$  and $($ respectively and the bits in $\bar{e}^R, \bar{c}^R$ will be encoded using $]$ and $)$ instead of $0$ and $1$ respectively.  
The numbers will be encoded as directed line graphs with $N$ edges\footnote{Recall that a graph $G = (V, E)$ is a line graph if the vertices $V$ can be arranged into a sequence $v_1, \dots, v_{|V|}$ such that all edges $e \in E$ are of the form $(v_i, v_{i+1})$.}. The label of the $i^{\text{th}}$ edge in the line graph corresponds to the $i^{\text{th}}$ bit in the encoding. For example, if $\bar{e} = 0011$, then we encode $\bar{e}$  as:
\begin{align*}
 L(e) = \blacksquare \xrightarrow[]{\text{[}}  \square \xrightarrow[]{\text{[}}  \square \xrightarrow[]{\text{(}}  \square \xrightarrow[]{\text{(}} \blacksquare
\end{align*}
and $\bar{e}^R = 1100$ as:
\begin{align*}
L^R(e) =  \blacksquare \xrightarrow[]{\text{)}}  \square \xrightarrow[]{\text{)}}  \square \xrightarrow[]{\text{]}}  \square \xrightarrow[]{\text{]}} \blacksquare
\end{align*}
We denote by $L_1 \circ L_2$ the stitching of the two line graphs, where the end node of $L_1$ becomes the start node of $L_2$. For example, $L(e) \circ L^R(e) =  \blacksquare \xrightarrow[]{\text{[}}  \square \xrightarrow[]{\text{[}}  \square \xrightarrow[]{\text{(}}  \square \xrightarrow[]{\text{(}} \blacksquare  \xrightarrow[]{\text{)}}  \square \xrightarrow[]{\text{)}}  \square \xrightarrow[]{\text{]}}  \square \xrightarrow[]{\text{]}} \blacksquare$ for edge $e$ with  $\bar{e} = 0011$.
%We can copy the vertex set of $H$ three times to parts $I, J, K$, and then plant the edges of $H$ to $I \times J, J \times K, K \times I$. Thus we only need to solve AE-Mono$\Delta$ on tripartite graphs. Say we want to report for every edge between parts $I$ and $J$, whether it is in a monochromatic triangle.

We now construct the input graph $G$ as follows. We start by creating five copies of the vertex set of $V_H$: $A, B, C, D, E$. We use $v_A, v_B, v_C, v_D, v_E$ to denote the copy of \changes{vertex} $v$ in $A,B,C,D,E$ respectively. For every edge $e = (u,v) \in E_H$, we then do the following:
\begin{itemize}
\item connect $u_A$ to $v_B$ with $L(e) \circ L(c(e))$;
\item connect $u_B$ to $v_C$ with $L^R(c(e)) \circ L(c(e))$; 
\item connect $u_C$ to $v_D$ with $L^R(c(e))$; and
\item connect $u_D$ to $v_E$ with $L^R(e)$.
\end{itemize}

We now execute $\cfga{\mathcal{D}_2}$ on $G$, \underline{\textbf{{\color{red}which has $O(n \log n)$ vertices}}}. To obtain a solution to AE-Mono$\Delta$, we simply filter the pairs $(u,v)$ such that $u \in A$ and $v \in E$: this can be done in time $O(n^2)$. Let $P$ be the resulting set. It remains to argue that the reduction is correct.

Suppose that $(u_A,v_E) \in P$. Then, there is a path from $u_A$ to $v_E$ in $G$ of the following form:
\begin{align*}
  u_A \xrightarrow[]{\bar{e} \circ \bar{c_1}}  w_B \xrightarrow[]{\bar{c_2}^R \circ \bar{c_2}}  t_C \xrightarrow[]{\bar{c_3}^R}  z_D \xrightarrow[]{\bar{e'}^R} v_E
\end{align*}
Since the labels of this path are recognized by Dyck-2, we must have that $c_1 = c_2$, $c_2 = c_3$, and $e = e'$. But $e = (u,w)$ and $e' = (z,v)$. Thus, $u=z$ and $v = w$. This implies that we have a triangle in $H$ that is formed by the nodes $u,v,t$. Moreover, the edges all have the same color, and hence the triangle is monochromatic. 

For the other direction, suppose that we have an edge $(u,v)$ that forms a monochromatic triangle $u,v,t$. It is easy to see by the above argument that $(u_A, v_E)$ will then appear in $P$. 
\end{proof}

The error in the proof is highlighted in red (underlined and in bold font). In particular, the claim that the graph $G$ has $O(n \log n)$ vertices is incorrect. Observe that the input graph $H$, which has $n$ nodes, can have as many as $n^2$ edges. Therefore, since we encode the color of each edge using a binary encoding via a line graph, the graph $G$ contains $O(n^2 \log n)$ vertices, leading only to a trivial lower bound of $\Omega(n^{{2.5/2} - \epsilon})$. The straightforward lower bound for the problem is $\Omega(n^2)$ since the all-pairs output can be as large as $n^2$. The problem of showing a higher lower bound (or a tighter upper bound) for all-pairs Dyck reachability, even when fast matrix multiplication is allowed, remains open.

\section{Missing Proofs}

\begin{figure}[!t]
	\begin{tikzpicture}
		\def\shift{2}
		\def\shiftC{4}
		\def\shiftD{6}
		
		\scalebox{1}{
			
			%% first CL1 gadget %%%%%%%%%%%%%%%%
			\node[label=left:$u$] (u)  at (-2, 0) { $\blacksquare$};
			\node (a1)  at (-1, 0) { $a_1$};
			\node (a2)  at (-1, -1){ $a_2$};
			\node (a3)  at (-1, -2) { $a_3$};
			\node (adots)  at (-1, -2.5) { $\vdots$};
			\node (a4)  at (-1, -3.7) { $a_n$};
			\node at (-1, -4.25) {\small $A$};
			
			\draw[line width=1.2,->] (u) -- (a1) node[midway,above]{$($} ;
			\draw[line width=1.2, ->] (a1) -- (a2) node[midway,left]{$($} ;
			\draw[line width=1.2, ->] (a2) -- (a3) node[midway,left]{$($} ;
			\draw[->] (adots) -- (a4) node[midway,left]{$($} ;
			
			\node  (b1)  at (-1+\shift, 0) { $b_1$};
			\node  (b2)  at (-1+\shift, -1){ $b_2$};
			\node (b3)  at (-1+\shift, -2) { $b_3$};
			\node (bdots)  at (-1+\shift, -2.5) { $\vdots$};
			\node  (b4)  at (-1+\shift, -3.7) { $b_n$};
			\node at (-1+\shift, -4.25) {\small $B$};
			
			\node (c1)  at (-1+\shiftC, 0) { $c_1$};
			\node (c2)  at (-1+\shiftC, -1){ $c_2$};
			\node (c3)  at (-1+\shiftC, -2) { $c_3$};
			\node (cdots)  at (-1+\shiftC, -2.5) { $\vdots$};
			\node  (c4)  at (-1+\shiftC, -3.7) { $c_n$};
			\node at (-1+\shiftC, -4.25) {\small $C$};
			
			\node[label=right:$a_1'$]  (d1)  at (-1+\shiftD, 0) { $\blacksquare$};
			\node (d2)  at (-1+\shiftD, -1){ $a_2'$};
			\node (d3)  at (-1+\shiftD, -2) { $a_3'$};
			\node (ddots)  at (-1+\shiftD, -2.5) { $\vdots$};
			\node (d4)  at (-1+\shiftD, -3.7) { $a_n'$};
			\node at (-1+\shiftD, -4.25) {\small $A'$};
			
			\draw[->] (d4) -- (ddots) node[midway,right]{$)$} ;
			\draw[line width=1.2,->] (d3) -- (d2) node[midway,right]{$)$} ;
			\draw[line width=1.2,->] (d2) -- (d1) node[midway,right]{$)$} ;
			
			\draw[->] (a1) -- (b1) node[midway,above]{$($} ;
			\draw[line width=1.2,->] (a3) -- (b2) node[midway,above]{$($} ;
			\draw[->] (a3) -- (b4) node[midway,above]{$($} ;
			\draw[line width=1.2,->] (b2) -- (c2) node[midway,above]{$)$} ;
			\draw[->] (b3) -- (c3) node[midway,above]{$)$} ;
			\draw[line width=1.2,->] (c2) -- (d3) node[midway,above]{$)$} ;
			\draw[->] (c2) -- (d1) node[midway,above]{$)$} ;
			\draw[->] (c3) -- (d4) node[midway,above]{$)$} ;
			
			}
	\end{tikzpicture}
	\caption{An example reduction for the proof of Theorem~\ref{thm:dyck:n}. The thick arrows show the valid Dyck-1 string that corresponds to the triangle $a_3 \rightarrow b_2 \rightarrow c_2 \rightarrow a_3$.} \label{fig:triangle}
\end{figure}
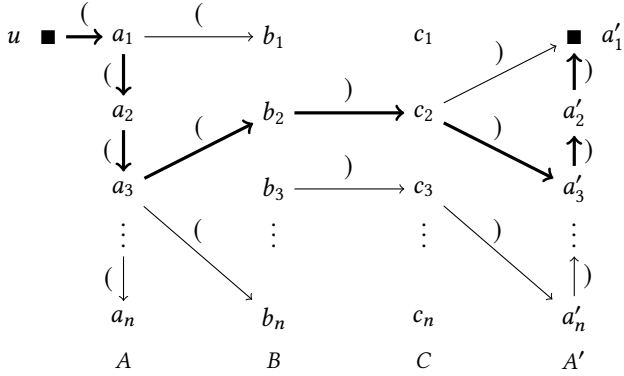

\cflod*

\begin{proof}
We will describe a reduction from the 3-Clique problem, which is subcubic-equivalent\footnote{This means that either both problems admit truly subcubic combinatorial algorithms or none.} to BMM. In particular, we can assume as an input a 3-partite undirected graph $G$ with partitions $A,B,C$ each of size $n$ and an edge set $E$. 
We will construct an input graph $H$ for the CFL reachability problem as follows (see also Figure~\ref{fig:triangle}). Our reduction uses only one type of open/close parentheses: $($ and $)$. (In fact, the hardness result here applies to an even "weaker" grammar, which is the one that produces the language $\{(^i )^i \mid i \geq 0\}$.)

Let $A = \{a_1, \dots, a_n\}$, $B = \{b_1, \dots, b_n\}$, and $C = \{c_1, \dots, c_n\}$. The vertex set for $H$ is $A \cup B \cup C$, plus a set 
of distinct fresh vertices $A' = \{a_1', \dots, a_n'\}$ plus one distinct fresh vertex $u$.
The edge set of $H$ is defined as follows:
\begin{align*}
  & \{ (u, a_1), (a_1,a_2), (a_2, a_3), \dots, (a_{n-1}, a_n)\} & \text{ with label } ( \\
%  & L(u_1, a_1), L(u_1, u_2), L(u_2, a_2), L(u_2, u_3), \dots, L(u_{n-1}, u_n), L(u_n, a_n) \\
  %& R(a_1', w_1), R(w_2, w_1), R(a_2', w_2), R(w_3, w_2), \dots, R(w_{n}, w_{n-1}), R(a_n', w_n) \\
  & \{ (a_i, b_j) \mid (a_i,b_j) \in E\} &\text{ with label } ( \\
  & \{ (b_i, c_j) \mid (b_i,c_j) \in E\} & \text{ with label } ) \\
  & \{ (c_i, a_j') \mid (c_i,a_j) \in E\} & \text{ with label } ) \\
    & \{ (a_n', a_{n-1}'), (a_{n-1}',a'_{n-2}),  \dots, (a_2', a_1') \} & \text{ with label } ) 
\end{align*}

Observe that our construction only adds $(n+1)$ additional vertices and $O(n)$ edges over the original graph $G$.

We will now show that the pair $(u, a_1')$ returns true for the \textsc{On-Demand} problem on the graph $H$ if and only if $G$ has a triangle. The main observation is that by following the unique path from $u$ to $a_i$ we see exactly $i$ edges with label $($. Similarly, by following the unique path from $a_i'$ to $a_1'$ we see exactly $(i-1)$ edges with label $)$. Finally, if we transition from any $a_i$ to any $a_j'$ we follow a path of length 3 where the labels are $())$. This means that the only way to follow a well-formed path from $u$ to $a_1'$ such that the string formed by its labels belongs to Dyck-$k$  is to transition from some $a_i$ to some $a_i'$; but this means that $G$ contains a triangle that uses the vertex $a_i$.
\end{proof}

\cubicn*

\begin{proof}
For the language $\{a^i s b^i \mid i \geq 0\}$, we follow the same construction as Theorem~\ref{thm:dyck:n}, with the only difference that $(i)$ each edge of the form $(b_i, c_j)$ is replaced by a (fresh and unique) path of length $|s|$ with labels from $s$, and $(ii)$ the on-demand pair is $(a_1, a_1')$. If $s = \emptyset$, we simply use $s = ab$.

\smallskip

The language with equal number of $a$'s and $b$'s is captured by the following CFG:
$$ S \gets aSbS \mid bSaS \mid ab \mid ba \mid aabb \mid abab \mid abba \mid baab \mid baba \mid bbaa$$
It is easy to see that the construction we used for Dyck-$1$ can be used for this CFG.

\smallskip

Finally, palindrome words of odd length are captured by the following CFG:
$$ S \gets \alpha_1 \mid \dots \mid \alpha_k \mid \alpha_1 S \alpha_1 \mid \dots \mid \alpha_k S \alpha_k $$
To do this construction, we keep the vertex set as the one in Theorem~\ref{thm:dyck:n} plus one more vertex $v$, but the edge set becomes:
\begin{align*}
  & \{  (u, a_1), (a_1,a_2), (a_2, a_3), \dots, (a_{n-1}, a_n)\}  & \text{with label } \alpha_1 \\
%  & L(u_1, a_1), L(u_1, u_2), L(u_2, a_2), L(u_2, u_3), \dots, L(u_{n-1}, u_n), L(u_n, a_n) \\
  %& R(a_1', w_1), R(w_2, w_1), R(a_2', w_2), R(w_3, w_2), \dots, R(w_{n}, w_{n-1}), R(a_n', w_n) \\
  & \{ (a_i, b_j) \mid (a_i,b_j) \in E\} & \text{with label } \alpha_2 \\
  & \{ (b_i, c_j) \mid (b_i,c_j) \in E\} & \text{with label } \alpha_2 \\
  & \{ (c_i, a_j') \mid (c_i,a_j) \in E\} & \text{with label } \alpha_2  \\
    & \{  (a_n', a_{n-1}'), (a_{n-1}',a'_{n-2}),  \dots, (a_2', a_1'), (a_1',v)\} & \text{with label } \alpha_1 
\end{align*}
Hence, every path from $u$ to $v$ will have labels that form a word of the form $\alpha_1 \dots \alpha_1 \alpha_2 \alpha_2 \alpha_2 \alpha_1 \dots \alpha_1$. It is easy to see that a triangle exists in $G$ if and only if the word is a palindrome (in particular, the number of $\alpha_1$ in the beginning and the end must be the same). For palindromes of even length the construction is similar.
\end{proof}

\medskip

\anbn*

\begin{proof}
We will use the fact that, under the combinatorial $k$-Clique hypothesis, for any constant $\epsilon >0$, for any $k> 2/\epsilon$, finding whether a ($k$-partite) directed graph contains a $k$-cycle can be detected in time $\Omega(m^{2-\epsilon})$~\cite{Lincoln020}.

We will prove the statement only for Dyck-1, since the construction is similar for the other cases.
We apply the same construction as in the proof of Theorem~\ref{thm:dyck:n}, with the only difference that what connects the vertices in $A$ and their copies in $A'$ is now a path of length $k$ (instead of a path of length 3, as in the case of triangles). Without any loss of generality, we can pick $k$ to be an odd number. Then, the first $\lfloor k/2 \rfloor$ edges in the path will have label $($, followed by $k- \lfloor k/2 \rfloor$ edges with label $)$. It is easy to see that, following the same argument as before, $(u,a_1')$ is a correct pair if and only if the graph $G$ has a $k$-cycle. Since the number of edges we added are $O(n)  = O(m)$, the desired result follows.
\end{proof}

\inversehomomorphism*

\begin{proof}
Let $\mathcal{G}$ be the CFG for $L$, and $\mathcal{G}'$ be the CFG for $L' = h^{-1}(L) = \{ w \mid h(w) \in L\}$. Consider an input graph $G$ for $\mathcal{G}$. We construct an input $G'$ for $\mathcal{G}'$ as follows.

For a symbol $\alpha$ in the alphabet of $L$, let $h(a) = \beta_1 \dots \beta_\ell$ where $\ell \geq 1$. Then, for any edge $(u,v) \in G$ with label $\alpha$, we introduce in $G$ the following path: 
$$ u  \xrightarrow[]{\beta_1} w_1\xrightarrow[]{\beta_2} w_2  \dots w_{\ell-1} \xrightarrow[]{\beta_\ell} v .$$ 
Here, $w_1, \dots, w_{\ell-1}$ are fresh distinct vertices. If $h(\alpha)$ is the empty word, we simply merge the nodes $u,v$ in $G'$. The new instance $G'$ has input size $O(m)$. We claim the following two statements: 

\smallskip
 
\noindent \textit{Claim 1: if $(u,v)$ is an output pair for $\mathcal{G}, G$ then $(u,v)$ is an output pair for $\mathcal{G}', G'$.}
Indeed, assume that $(u,v)$ is in the output of $G$. Then, there is a path from $u$ to $v$ in $G$ with labels that form some word $w \in L$. From our construction, $G'$ contains a path from $u$ to $v$ with labels that form the word $h(w)$. Hence, $(u,v)$ is in the output of $G'$.

\smallskip

\noindent \textit{Claim 2: if $(u,v)$ is an output pair for $\mathcal{G}', G'$ and $u,v$ occur in $G$, then $(u,v)$ is an output pair for $\mathcal{G}, G$.}
Indeed, assume that $(u,v)$ is in the output of $G'$. Then, there is a path from $u$ to $v$ in $G'$ with labels that form a word $w' \in h^{-1}(L)$. Since $u,v$ are vertices in $G$, by our construction there must be a path in $G$ from $u$ to $v$ with labels that form a word $w$ such that $w' = h(w)$. But then, for $w$ we have the property that $h(w) \in h^{-1}(L)$, and thus $w \in L$.
 
\smallskip

Now, the algorithm for the \textsc{On-Demand} problem with input $u,v$ creates the instance $G'$ and then runs $\mathcal{G}'$ on $G'$ to check whether $(u,v)$ is true. This has running time $O(m^c)$. 

For the  \textsc{All-Pairs} problem, we also create the instance $G'$ and then run $\mathcal{G}'$ on $G'$. This takes time $O(m^c)$. Then, we need to filter out the outputs that have constants not in $G$, which takes time linear in the size of the output. Since the output is of size at most $O(m^c)$, the claim follows.
\end{proof}

\section{Undecidability Proofs}

In this section, we prove the results on the undecidability of classifying CFGs  in terms of their data complexity using Greibach's theorem. All proofs use the same technique with small variations.

\undecidableone*

\begin{proof}
Fix a constant $c \in [2,3)$. To prove undecidability, we apply Greibach's theorem. Consider the following property $C$ for a CFL: any CFG that produces the language can be evaluated in time $O(m^c)$ by a combinatorial algorithm. As we have seen, $C$ is satisfied by all regular languages (since every regular language can be produced by a linear CFG, and a linear CFG can be evaluated in time $O(m^2) = O(m^c)$). It is also non-trivial, since under the combinatorial $k$-Clique hypothesis, Dyck-$k$ cannot be evaluated in time $O(m^{3-\epsilon})$ for any constant $\epsilon >0$ by a combinatorial algorithm, hence it does not admit an $O(m^c)$ combinatorial algorithm. It remains to show that $C$ is closed under the right quotient by a single symbol. 

\smallskip

Indeed, take a CFL $\mathcal{L}$ and a corresponding CFG $\mathcal{G}$. Consider the language $\mathcal{L} / \alpha$ for a single symbol $\alpha$. We now want to evaluate the CFG $\mathcal{G}_\alpha$ that corresponds to the language $\mathcal{L} / \alpha$. To do this, we extend the input graph $G$ ($m = |E|$) as follows: for each possible vertex $v$, we add an edge $(v,t_v)$ with label $\alpha$, where $t_v$ is a fresh distinct vertex. Let $G'$ be the resulting graph. Note that $|E'| = |E| + n = O(m)$.
Then, we run the algorithm for $\mathcal{G}$ on the new instance $G'$, which runs in time $O(m^c)$. 
Finally, we can see that by construction, $(u,v)$ is an output pair for $\mathcal{G}_\alpha$ if and only if $(u,t_v)$ is an output tuple for $\mathcal{G}$. Hence, to obtain the output for $\mathcal{G}_\alpha$ it remains to do the following: for every pair of the form $(u,t_v)$, output $(u,v)$. This can be done in time $O(n^2) = O(m^2)$ by iterating over all output pairs. 
\end{proof}

\undecidable*

\begin{proof}
Fix a constant $c \in [2,3)$. To prove undecidability, we apply again Greibach's theorem. Consider the following property $C$ for a CFL: the \textsc{On-Demand} problem for any CFG that produces the language can be evaluated in time $O(n^c)$ by a combinatorial algorithm. As we have seen, $C$ is satisfied by all regular languages. It is also non-trivial, since under the combinatorial BMM hypothesis, $\{a^i b^i \mid i \geq 0\}$ cannot be evaluated in time $O(n^{3-\epsilon})$ for any constant $\epsilon >0$ by a combinatorial algorithm, hence it does not admit an $O(n^c)$ combinatorial algorithm. It remains to show that $C$ is closed under the right quotient by a single symbol. 

\smallskip

Indeed, take a CFL $\mathcal{L}$ and a corresponding a corresponding CFG $\mathcal{G}$. Consider the language $\mathcal{L} / \alpha$ for a single symbol $\alpha$. We now want to evaluate the \textsc{On-Demand} problem for a CFG $\mathcal{G}_\alpha$ that produces the language $\mathcal{L} / \alpha$. Let $(s,t)$ be the input pair. To do this, we extend the input graph $G$ as follows: we add an edge $\alpha(t,t')$ to the instance, where $t'$ is a fresh distinct vertex. Let $G'$ be the resulting instance.
Then, we run the algorithm for $\mathcal{G}$ on the new instance $G'$, which runs in time $O((n+1)^c) = O(n^c)$. 
Finally, we can see that by construction, $(s,t)$ is an output tuple for $\mathcal{G}_\alpha$ if and only if $(s,t')$ is an output tuple for $\mathcal{G}$. 
\end{proof}

\undecidablethree*

\begin{proof}
Fix a constant $c \in [1,3)$. To prove undecidability, we apply again Greibach's theorem. Consider the following property $C$ for a CFL: the \textsc{On-Demand} problem for any CFG that produces the language can be evaluated in time $O(m^c)$ by a combinatorial algorithm. As we have seen, $C$ is satisfied by all regular languages, since the on-demand problem can be evaluated in linear time. It is also non-trivial, since under the $k$-Clique hypothesis, Dyck-2 cannot be evaluated in time $O(m^{3-\epsilon})$ for any constant $\epsilon >0$ by a combinatorial algorithm, hence it does not admit an $O(m^c)$ combinatorial algorithm. It remains to show that $C$ is closed under the right quotient by a single symbol. 

\smallskip

Indeed, take a CFL $\mathcal{L}$ and a corresponding CFG $\mathcal{G}$. Consider the language $\mathcal{L} / \alpha$ for a single symbol $\alpha$. We now want to evaluate the \textsc{On-Demand} problem for the CFG $\mathcal{G}_\alpha$ that produces the language $\mathcal{L} / \alpha$. Let $(s,t)$ be the input pair. To do this, we extend the input instance $I$ as follows: we add an edge $\alpha(t,t')$ to the instance, where $t'$ is a fresh distinct value. Let $G'$ be the resulting instance.
Then, we run the algorithm for $\mathcal{G}$ on the new instance $G'$, which runs in time $O((m+1)^c) = O(m^c)$. 
Finally, we can see that by construction, $(s,t)$ is an output tuple for $\mathcal{G}_\alpha$ if and only if $(s,t')$ is an output tuple for $\mathcal{G}$. 
\end{proof}

\section{Remaining Proof for APA Lower Bound}
\label{sec:apa}

We start with the following simple observations.

\begin{proposition} \label{prop:use}
Every valid word starts with $\alpha$. Moreover, every word with length at least 2 ends with either $\beta$ or $\bar \alpha$.
\end{proposition}

\begin{proposition} \label{prop:simple}
The following productions are valid:
\begin{itemize} 
\item $T \gets \alpha \gamma T \bar \gamma \bar \alpha$
\item $T \gets T \gamma T \bar \gamma \bar \alpha$
\end{itemize}
\end{proposition}

Note that $\alpha$ is the only word (of length 1) that ends with a symbol that is not $\beta, \bar \alpha$.

\begin{claim} \label{claim:1}
	If the \textsc{On-Demand} problem on Andersen's analysis returns true, then there exists a $3k$-clique in the input graph.
\end{claim} 

\begin{proof}
Let $w$ be the word that forms on the path from $p$ to $q$. In particular, $w$ is of the form:
\begin{align*} 
w  & =   \alpha \{ L(v_1) \dots L(v_k) \} \{ L(w_1) \dots L(w_k) \} \alpha \{ L^R(w_k') \dots L^R(w_1')\}  \\
& \quad\quad\quad \gamma \{ L(z_1) \dots L(z_k) \} \alpha \{ L^R(z_k') \dots L^R(z_1')\} \bar \gamma \bar \alpha \{ L^R(v_k') \dots L^R(v_1')\} \beta \\
 & =   \alpha \{ \alpha^{v_1} \alpha \gamma \dots \alpha^{v_k} \alpha \gamma \} \{ \alpha^{w_1} \alpha \gamma \dots \alpha^{w_k} \alpha \gamma \} \alpha \{ \bar \gamma \bar \alpha \beta^{w_k'} \dots \bar \gamma \bar \alpha \beta^{w_1'} \}  \\
 & \quad\quad\quad \gamma \{ \alpha^{z_1} \alpha \gamma \dots \alpha^{z_k} \alpha \gamma \} \alpha \{  \bar \gamma \bar \alpha \beta^{z_k'}\dots  \bar \gamma \bar \alpha \beta^{z_1'}\} \bar \gamma \bar \alpha \{  \bar \gamma \bar \alpha \beta^{v_k'} \dots  \bar \gamma \bar \alpha \beta^{v_1'}\} \beta
\end{align*}

First, note that $w$ ends with $\beta$. This means that $w$ was generated by the rule $T \gets TT \beta$. From Proposition~\ref{prop:use} and our construction, the first $T$ can only match the first $\alpha$ of this word. Thus, the following word is also valid:
\begin{align*} 
 & \{ \alpha^{v_1} \alpha \gamma \dots \alpha^{v_k} \alpha \gamma \} \{ \alpha^{w_1} \alpha \gamma \dots \alpha^{w_k} \alpha \gamma \} \alpha \{ \bar \gamma \bar \alpha \beta^{w_k'} \dots \bar \gamma \bar \alpha \beta^{w_1'} \} \\
  & \quad\quad\quad   \gamma \{ \alpha^{z_1} \alpha \gamma \dots \alpha^{z_k} \alpha \gamma \} \alpha \{  \bar \gamma \bar \alpha \beta^{z_k'}\dots  \bar \gamma \bar \alpha \beta^{z_1'}\} \bar \gamma \bar \alpha \{  \bar \gamma \bar \alpha \beta^{v_k'} \dots  \bar \gamma \bar \alpha \beta^{v_1'}\}
\end{align*}
We repeat this process $v_1'$ more times. Observe that if $v_1' > v_1$, the first $T$ would not be able to match, since the word would start with $\gamma$, which is not valid. Hence, $v_1' \leq v_1$. We are now left with the following word: 
\begin{align} \label{eq:one}
&  \{ \alpha^{v_1-v_1'} \alpha \gamma \dots \alpha^{v_k} \alpha \gamma \} \{ \alpha^{w_1} \alpha \gamma \dots \alpha^{w_k} \alpha \gamma \} \alpha \{ \bar \gamma \bar \alpha \beta^{w_k'} \dots \bar \gamma \bar \alpha \beta^{w_1'} \} \nonumber \\ 
 & \quad\quad\quad  \gamma \{ \alpha^{z_1} \alpha \gamma \dots \alpha^{z_k} \alpha \gamma \} \alpha \{  \bar \gamma \bar \alpha \beta^{z_k'}\dots  \bar \gamma \bar \alpha \beta^{z_1'}\} \bar \gamma \bar \alpha \{  \bar \gamma \bar \alpha \beta^{v_k'} \dots  \bar \gamma \bar \alpha\}
\end{align}
Since the above word ends with $\bar \alpha$, it must have been generated by the rule $T \gets T \gamma \bar T$. We will refer to the central $\gamma$ in~\autoref{eq:one} as $\gamma^\star$. First, note that no $\gamma$ to the left of $\gamma^\star$ (other than the very first $\gamma$) can act as the separator for another application of the rule $T \gets T \gamma \bar T$. This is because every $\gamma$ is preceded by an $\alpha$, which would lead $T$ to end with an $\alpha$ and thus violates  \autoref{prop:use}. Similarly, no $\gamma$ to the right of $\gamma^\star$ can also lead to valid parsing. We now argue that $\gamma^\star$ is also an invalid choice. Consider the following $\bar T$:
$$ \bar{T} =  \{ \alpha^{z_1} \alpha \gamma \dots \alpha^{z_k} \alpha \gamma \} \alpha \{  \bar \gamma \bar \alpha \beta^{z_k'}\dots  \bar \gamma \bar \alpha \beta^{z_1'}\} \bar \gamma \bar \alpha \{  \bar \gamma \bar \alpha \beta^{v_k'} \dots  \bar \gamma \bar \alpha \beta^{v_1'}\}$$

The inverse of this string can only be parsed by $T \gets T T \beta$ (since there is no valid choice of $\gamma$ as every $\gamma$ is preceded by an $\alpha$), and the parsing forces $z_1 = v'_1, z_2 = v'_2 $ and so on. Eventually, we will need to parse $T$ that starts with $\bar \gamma \bar \alpha$ which violates \autoref{prop:use}. We have now established that the only choice of $\gamma$ is the very first one in \autoref{eq:one} which implies that $T$ is equal to $\alpha^{v_1-v_1'} \alpha$. But this word is valid only if it is of length one, hence it must be that $v_1' = v_1$. Now, we are left with the following word that is in $\bar{T}$:
\begin{align*} 
 & \{  \alpha^{v_2} \alpha \gamma  \dots \alpha^{v_k} \alpha \gamma \} \{ \alpha^{w_1} \alpha \gamma \dots \alpha^{w_k} \alpha \gamma \} \alpha \{ \bar \gamma \bar \alpha \beta^{w_k'} \dots \bar \gamma \bar \alpha \beta^{w_1'} \}  \\
 & \quad\quad\quad   \gamma \{ \alpha^{z_1} \alpha \gamma \dots \alpha^{z_k} \alpha \gamma \} \alpha \{  \bar \gamma \bar \alpha \beta^{z_k'}\dots  \bar \gamma \bar \alpha \beta^{z_1'}\} \bar \gamma \bar \alpha \{  \bar \gamma \bar \alpha \beta^{v_k'} \dots  \bar \gamma \bar \alpha\}
\end{align*}
Since the inverse of this word ends in $\bar \gamma \bar \alpha$, the only way to generate it is by $T \rightarrow T \gamma \bar T \rightarrow  T \bar \gamma \bar \alpha$. This leaves us with the following word in $T$:
\begin{align*} 
 & \{  \alpha^{v_2} \alpha \gamma  \dots \alpha^{v_k} \alpha \gamma \} \{ \alpha^{w_1} \alpha \gamma \dots \alpha^{w_k} \alpha \gamma \} \alpha \{ \bar \gamma \bar \alpha \beta^{w_k'} \dots \bar \gamma \bar \alpha \beta^{w_1'} \}  \\
  & \quad\quad\quad  \gamma \{ \alpha^{z_1} \alpha \gamma \dots \alpha^{z_k} \alpha \gamma \} \alpha \{  \bar \gamma \bar \alpha \beta^{z_k'}\dots  \bar \gamma \bar \alpha \beta^{z_1'}\} \bar \gamma \bar \alpha \{  \bar \gamma \bar \alpha \beta^{v_k'} \dots  \bar \gamma \bar \alpha \beta^{v_2'} \}
\end{align*}
We now repeat the same logic $(k-1)$ more times and obtain that $v_2 = v_2', \dots, v_k = v'_k$. At this point, we are left with the following word that is in $T$:
\begin{align*} 
 \{ \alpha^{w_1} \alpha \gamma \dots \alpha^{w_k} \alpha \gamma \} \alpha \{ \bar \gamma \bar \alpha \beta^{w_k'} \dots \bar \gamma \bar \alpha \beta^{w_1'} \}  \gamma \{ \alpha^{z_1} \alpha \gamma \dots \alpha^{z_k} \alpha \gamma \} \alpha \{  \bar \gamma \bar \alpha \beta^{z_k'}\dots  \bar \gamma \bar \alpha \beta^{z_1'}\} \bar \gamma \bar \alpha 
\end{align*}
Since this word ends with $\bar \alpha$, it must be generated by the rule $T \gets T \gamma \bar T$. The only valid way to do this production is that $\gamma$ corresponds to $\gamma^\star$ (the central $\gamma$) by following the derivation shown above. This means that the following word is recognized by the grammar:
\begin{align*} 
 \{ \alpha^{w_1} \alpha \gamma \dots \alpha^{w_k} \alpha \gamma \} \alpha \{ \bar \gamma \bar \alpha \beta^{w_k'} \dots \bar \gamma \bar \alpha \beta^{w_1'} \} 
\end{align*}
 For this case, we can use the same logic as above to show that $w_1 = w_1', \dots, w_k = w_{k}'$. Also, we have the inverse of the following word is recognized by the grammar:
\begin{align*} 
\{ \alpha^{z_1} \alpha \gamma \dots \alpha^{z_k} \alpha \gamma \} \alpha \{  \bar \gamma \bar \alpha \beta^{z_k'}\dots  \bar \gamma \bar \alpha \beta^{z_1'}\} \bar \gamma \bar \alpha 
\end{align*}
For this to happen, the following word must be in $T$:
\begin{align*} 
\{ \alpha^{z_1} \alpha \gamma \dots \alpha^{z_k} \alpha \gamma \} \alpha \{  \bar \gamma \bar \alpha \beta^{z_k'}\dots  \bar \gamma \bar \alpha \beta^{z_1'}\} 
\end{align*}
This implies that $z_1 = z_1', \dots, z_k = z_{k}'$. We have now established that there exist $v_1, \dots, v_k$, $w_1, \dots w_k$, $z_1, \dots z_k$ that satisfy the grammar. From the gadget construction, these vertices correspond to three  $k$-cliques. Using the same argument from the hardness proof of Dyck-2 in conjunction with \autoref{obv:1}, it is now straightforward to establish that the three $k$-cliques are also disjoint.
\end{proof}

\begin{claim} 
	If there exists a $3k$-clique in the input graph, then the \textsc{On-Demand} problem on Andersen's analysis returns true.
\end{claim}

\begin{proof}
	Let $t_1, t_2, t_3 \in \mC_k$ be three disjoint $k$-cliques. We will show that there exists a path from $u$ to $v$ that forms a valid word.
	Consider the path formed by the vertices $CL(t_1), CNG_2(t_2), \overline{CL}_2(t_2)$, $CNG_2(t_3), \overline{CL}_3(t_3), \overline{CNG}_3(t_1)$. Let $t_1 = \{v_1, \dots, v_k\}$, $t_2 = \{w_1, \dots, w_k\}$ and $t_3 = \{z_1, \dots, z_k\}$. By applying the rule $T \gets TT\beta$ $(v_1+1)$ times, we obtain $T^{v_1+1} T \beta^{v_1+1}$. Now, we apply $T \gets \alpha$ to the first $(v_1+1)$ occurrences of $T$ to obtain $\alpha \alpha^{v_1} T \beta^{v_1} \beta$. By Proposition~\ref{prop:simple}, we obtain $\alpha \alpha^{v_1} \alpha \gamma T \bar \gamma \bar \alpha \beta^{v_1} \beta$. Since $L(v_1) = \alpha^{v_1} \alpha \gamma$ and $L^R(v_1) = \bar \gamma \bar \alpha \beta^v$, we have so far generated $\alpha  L(v_1) T  L^R(v_1) \beta$. We now repeat this process $k-1$ more times for $T$ to generate the following word:
$$ \alpha  L(v_1) L(v_2)  \dots L(v_k) T  L^R(v_k) \dots L^R(v_2) L^R(v_1) \beta $$	
 	
Now, from Proposition~\ref{prop:simple} we obtain:	
$$ \alpha  L(v_1) L(v_2)  \dots L(v_k)  T \gamma T \bar \gamma \bar \alpha  L^R(v_k) \dots L^R(v_2) L^R(v_1) \beta $$
Finally, we use the same construction as above for each of the two $T$'s to generate the following final word:
\begin{align*} 
& \alpha \{ L(v_1) \dots L(v_k) \} \{ L(w_1) \dots L(w_k) \} \alpha \{ L^R(w_k) \dots L^R(w_1)\} \\
 & \quad\quad\quad \gamma \{ L(z_1) \dots L(z_k) \} \alpha \{ L^R(z_k) \dots L^R(z_1)\} \bar \gamma \bar \alpha \{ L^R(v_k) \dots L^R(v_1)\} \beta 
\end{align*}
One can observe that this word matches the labels of the path we considered in the beginning.
\end{proof}

\bibliography{refs}

\end{document}